\documentclass[11pt]{article}
\usepackage{fullpage}
\usepackage[hyphens]{url}
\usepackage{amsmath,amssymb,amsthm,mathtools}
\usepackage{dsfont}
\usepackage{enumitem}
\usepackage{booktabs}
\usepackage{array}
\usepackage{microtype}
\usepackage{caption}
\usepackage{xcolor}
\usepackage{thmtools}
\usepackage{thm-restate}
\usepackage[
  date=year,
  eprint=false,
  doi=false,
  isbn=false,
  backend=biber,
  giveninits=true,
  uniquename=init,
  maxcitenames=2,
  maxbibnames=99,
  natbib=true,
  url=false,
  sorting=ynt,
  style=apa,
  citestyle=authoryear-comp,
  backref=true,
  uniquelist=false
]{biblatex}

\DeclareSourcemap{
  \maps[datatype=bibtex]{
    \map{
      \step[fieldset=editor, null]
      \step[fieldset=series, null]
      \step[fieldset=language, null]
      \step[fieldset=address, null]
      \step[fieldset=location, null]
      \step[fieldset=month, null]
      \step[fieldset=annote, null]
    }
  }
}
\DefineBibliographyStrings{english}{%
  backrefpage  = {cited on page},
  backrefpages = {cited on pages},
}
\DefineBibliographyStrings{american}{%
  backrefpage  = {cited on page},
  backrefpages = {cited on pages},
}
\DefineBibliographyStrings{american-apa}{%
  backrefpage  = {cited on page},
  backrefpages = {cited on pages},
}
\renewbibmacro*{pageref}{%
  \iflistundef{pageref}
    {}
    {\printtext[parens]{%
       \midsentence%
       \ifnumgreater{\value{pageref}}{1}
         {\bibstring{backrefpages}\ppspace}
         {\bibstring{backrefpage}\ppspace}%
       \printlist[pageref][-\value{listtotal}]{pageref}}}}
\usepackage[colorlinks=true,linkcolor=blue!45!black,citecolor=blue!45!black,urlcolor=blue!55!black]{hyperref}
\usepackage[nameinlink,capitalize]{cleveref}
\hypersetup{
  pdftitle={Sensitivity and Differential Privacy in Metric Voting with Distortion below Three},
  pdfauthor={Shinsaku Sakaue, Kaito Fujii, Soh Kumabe, Yuichi Yoshida}
}
\newtheorem{theorem}{Theorem}[section]
\newtheorem{lemma}[theorem]{Lemma}
\newtheorem{proposition}[theorem]{Proposition}
\newtheorem{corollary}[theorem]{Corollary}

\theoremstyle{definition}

\theoremstyle{remark}
\newtheorem{remark}[theorem]{Remark}

\crefname{theorem}{Theorem}{Theorems}
\crefname{lemma}{Lemma}{Lemmas}
\crefname{proposition}{Proposition}{Propositions}
\crefname{corollary}{Corollary}{Corollaries}
\crefname{claim}{Claim}{Claims}
\crefname{definition}{Definition}{Definitions}
\crefname{example}{Example}{Examples}
\crefname{remark}{Remark}{Remarks}
\crefname{appendix}{Appendix}{Appendices}
\Crefname{theorem}{Theorem}{Theorems}
\Crefname{lemma}{Lemma}{Lemmas}
\Crefname{proposition}{Proposition}{Propositions}
\Crefname{corollary}{Corollary}{Corollaries}
\Crefname{claim}{Claim}{Claims}
\Crefname{definition}{Definition}{Definitions}
\Crefname{example}{Example}{Examples}
\Crefname{remark}{Remark}{Remarks}
\Crefname{appendix}{Appendix}{Appendices}
\newcommand{\Cands}{\mathcal C}
\newcommand{\Voters}{\mathcal V}
\newcommand{\Lists}{\mathcal L}
\newcommand{\X}{\mathcal X}
\newcommand{\Xsn}{\mathcal X^{\mathsf{sn}}}
\newcommand{\E}{\mathbb E}
\newcommand{\Prob}{\mathbb P}
\newcommand{\SC}{\operatorname{SC}}
\newcommand{\OPT}{\operatorname{OPT}}
\newcommand{\Dist}{\operatorname{Dist}}
\newcommand{\AS}{\operatorname{AS}}
\newcommand{\Sens}{\operatorname{Sens}}
\newcommand{\RD}{\operatorname{RD}}
\newcommand{\SQ}{\operatorname{SQ}}

\let\set\relax
\DeclarePairedDelimiter{\set}{\{}{\}}
\let\Set\relax
\DeclarePairedDelimiterX{\Set}[2]{\{}{\}}{\,{#1}\,:\,{#2}\,}
\let\abs\relax
\DeclarePairedDelimiter{\abs}{|}{|}
\DeclareMathOperator*{\Supp}{\mathrm{Supp}}
\newcommand{\argmin}{\operatorname*{argmin}}
\newcommand{\consistent}{\rhd}
\newcommand*\diff{\mathop{}\!\mathrm{d}}
\newcommand{\W}{\mathrm W}
\newcommand{\e}{\mathrm e}
\newif\ifmainbodylabels
\mainbodylabelstrue

\newcommand{\mainbodyequationlabel}[1]{\ifmainbodylabels\label{#1}\else\notag\fi}
\title{Sensitivity and Differential Privacy in Metric Voting \\ with Distortion below Three}
\author{%
Shinsaku Sakaue%
\thanks{CyberAgent, Tokyo, Japan; National Institute of Informatics, Tokyo, Japan; Center for Advanced Intelligence Project, RIKEN, Tokyo, Japan. Email: \texttt{shinsaku.sakaue@gmail.com}.}%
\and
Kaito Fujii%
\thanks{Kyoto University, Kyoto, Japan. Email: \texttt{fujii@i.kyoto-u.ac.jp}.}%
\and
Soh Kumabe%
\thanks{CyberAgent, Tokyo, Japan. Email: \texttt{kumabe\_soh@cyberagent.co.jp}.}%
\and
Yuichi Yoshida%
\thanks{National Institute of Informatics, Tokyo, Japan. Email: \texttt{yyoshida@nii.ac.jp}.}%
}
\date{}
\begin{document}
\maketitle
\begin{abstract}
Voting rules aggregate individual preferences into collective decisions, but the rankings they receive contain only ordinal information.
The metric distortion framework studies ordinal voting rules in settings where voters and candidates are embedded in an unknown metric space.  Deterministic rules have optimal worst-case distortion $3$, while recent randomized rules break the $3$ barrier.  We study whether such improvements can coexist with low worst-case sensitivity with respect to the Wasserstein distance of lotteries under one-voter deletion and approximate differential privacy under one-voter replacement.  On the sensitivity side, we give a randomized rule with distortion at most $3-\varepsilon$ for an absolute constant $\varepsilon>0$ and, for $m$ candidates and $n$ voters, a worst-case sensitivity bound of $O((\log m+1)/n)$.  On the privacy side, for every $\delta\in(0,1)$ and all $n$ above an absolute constant, we construct a variant rule whose mechanism releasing a single sampled winner has distortion at most $3-\varepsilon$ and is $(O((\log m+\log(1/\delta)+1)/n),\delta)$-differentially private.  Both constructions use the same family of Gibbs distributions over constant-size candidate lists, with only the temperature parameter differing between the sensitivity and differential-privacy guarantees.  Our analysis builds on the biased-metric viewpoint behind the recent improvement over the $3$ barrier and proves a stability property for the biased-metric ratio.
\end{abstract}
\section{Introduction}
Voting rules aggregate individual preferences into collective decisions.  In single-winner elections, a rule maps the voters' rankings to one candidate or to a lottery over candidates.  Since rankings reveal only relative preferences, a basic question is how well an ordinal rule can approximate a socially efficient outcome when the underlying utilities or costs are hidden.  The distortion framework of \citet{ProcacciaRosenschein2006} formalizes this information loss by comparing the selected outcome with the optimum under the hidden utilities or costs; its metric version, developed by \citet{ABEPS2018}, specializes the comparison to metric costs.  There is a set of voters and a set of candidates embedded in a common metric space, and the social objective is to choose a candidate of small total distance to the voters.  The rule, however, sees only the ordinal rankings induced by the hidden metric.  The worst-case approximation factor under this information constraint is the rule's \emph{metric distortion}.  A long line of work, surveyed for example by \citet{AFSV2021}, studied lower bounds, fairness properties, deterministic constructions, and limited-information variants~\citep{GoelKrishnaswamyMunagala2017,MunagalaWang2019,AnagnostidesFotakisPatsilinakos2022,EbadianKahngPetersShah2022}.  It culminated in the resolution of the deterministic case: distortion $3$ is achievable by \citet{GHS2020} and optimal by the lower bound of \citet{ABEPS2018}.  Randomization improves the landscape.  With $n$ voters, Random Dictatorship has distortion $3-2/n$~\citep{AnshelevichPostl2017}; lower bounds above $2$ for randomized rules and structural progress were obtained by \citet{CharikarRamakrishnan2022} and \citet{PulyassarySwamy2021}.  \citet{CRWW2024} then broke the $3$ barrier for randomized rules by combining Maximal Lotteries with complementary rules, obtaining distortion at most $2.753$.  More recently, \citet{Cai2026} showed that bounded randomness already suffices to beat $3$: uniformly sampling the winner from a deterministically chosen constant-size multiset of candidates achieves distortion $3-\varepsilon$ for an absolute constant $\varepsilon>0$.
\begin{table}[t]
\centering
\small
\setlength{\tabcolsep}{4pt}
\renewcommand{\arraystretch}{1.15}
\begin{tabular}{p{0.25\textwidth}p{0.1\textwidth}p{0.26\textwidth}p{0.3\textwidth}}
\toprule
Statement & Distortion & Sensitivity & $(\varepsilon_{\mathsf{DP}},\delta)$-differential privacy \\
\midrule
Proposed Gibbs rule $Q^{(1)}$ & $\le 3-\varepsilon_{\mathsf{dist}}$ & $O((\log m+1)/n)$ & --- \\
Proposed Gibbs rule $Q_\delta$ & $\le 3-\varepsilon_{\mathsf{dist}}$ & $O((\log m+\log(1/\delta)+1)/n)$ & $(O((\log m+\log(1/\delta)+1)/n),\delta)$ \\
Random Dictatorship & $\le 3-2/n$ & $O(1/n)$ only on average & $(0,1/n)$ \\
Kempe's randomized rule & $\le 3-2/m$ & $O(1/n)$ only on average & $(0,\min\{1,9/n\})$ \\
Deterministic lower bounds & finite & $\Omega(1)$ worst-case for two candidates & $\delta\ge1$ (no nontrivial guarantee) \\
Universal lower bounds & finite & $\Omega(1/n)$ for two candidates & $\delta=\Omega(1/n)$ if $\varepsilon_{\mathsf{DP}}=O(1/n)$ \\
\bottomrule
\end{tabular}
\caption{Main guarantees, baselines, and lower bounds.  Here, $\varepsilon_{\mathsf{dist}}>0$ is an absolute distortion gap. Privacy entries concern the released winner, not the full lottery; the formal $u=1$ privacy guarantee for $Q^{(1)}$ has only the constant value $\delta=\e^{-1}$ and is omitted.  The baseline rules have $O(1/n)$ average sensitivity, but their worst-case sensitivity is $\Omega(1)$; their distortion bounds are size-fixed and their improvements over $3$ vanish with $n$ or $m$.
Lower bounds assume finite distortion.  For a subclass of deterministic rules, an $\Omega(1)$ average-sensitivity lower bound also holds.
Details of baseline results and lower bounds are in \Cref{app:baselines,app:universal-lower}.}
\label{tab:summary-results}
\end{table}
Alongside metric distortion, low sensitivity and privacy are important desiderata for voting rules.
On the privacy side, differential privacy~\citep{DworkMcSherryNissimSmith2006,DworkRoth2014} has been widely studied in various social choice settings~\citep{LiuLuXiaZikas2020,Torra2020,AlabiGhaziKumarManurangsi2022,TaoChenXuShi2022,LiLiuXiaCaoWang2023}.
For sensitivity, we ask how much the candidate lottery can change when a single voter is absent.
Worst-case sensitivity (or simply sensitivity for short), studied for randomized matching algorithms by \citet{YoshidaZhou2021}, measures the movement of an output distribution after a deletion of one input element using the Wasserstein distance.  We adapt this idea to metric voting to quantify changes in the candidate lottery.  Beyond its intrinsic interest, sensitivity in voting has another motivation in machine learning.
Recent work on benchmarking uses social-choice ideas to aggregate multi-task model evaluations, viewing models as candidates and benchmark tasks as voters~\citep{RofinEtAl2023,ZhangHardt2024}.  This perspective, although indirect, motivates a relevant sensitivity question: how much can the aggregate distribution change when one benchmark task is absent? For further discussion of related work, see \Cref{app:related-work}.

The above background motivates the following question:
\begin{quote}
\textit{Can randomized voting rules keep distortion below $3$ while satisfying low sensitivity and differential privacy?
   }
\end{quote}
\subsection{Our Contributions}
\paragraph{Main results.}  Let $m$ and $n$ denote the number of candidates and voters, respectively.  Our sensitivity result (\Cref{thm:main}) constructs a randomized ordinal rule $Q$ with distortion below $3$ by an absolute constant and worst-case sensitivity $O((\log m+1)/n)$ under one-voter deletion.  The sensitivity bound extends to multiple voter deletions, giving $O(k(\log m+1)/n)$ when $k\le n/2$ voters are deleted; see \Cref{cor:k-deletion-worst-case-sensitivity}.  For privacy, \Cref{thm:main-privacy} shows that, for every $\delta\in(0,1)$ and all sufficiently large $n$, there is a randomized ordinal rule $Q_\delta$ with distortion below $3$ by an absolute constant such that sampling and releasing a winner from its candidate lottery is $(O((\log m+\log(1/\delta)+1)/n),\delta)$-differentially private under replacement of one voter's ranking; specifically, setting $\delta=1/n$ gives $(O((\log m+\log n+1)/n),1/n)$-differential privacy.  The same rule $Q_\delta$ satisfies the corresponding sensitivity bound of $O((\log m+\log(1/\delta)+1)/n)$ by our general result in \Cref{thm:gibbs-over-lists}.
\Cref{tab:summary-results} summarizes these guarantees together with the baseline comparisons and lower bounds discussed below.
The construction is informational rather than computationally efficient; see \Cref{app:limitations} for limitations and open directions.

\paragraph{Lower bounds.}  To complement the main upper bounds, \Cref{app:universal-lower} gives the lower bounds summarized in \Cref{tab:summary-results}.
Every deterministic finite-distortion rule has $\Omega(1)$ worst-case sensitivity already on two-candidate profiles.  For a subclass of standard deterministic rules, an $\Omega(1)$ bound also holds for average sensitivity.  For differential privacy, no deterministic finite-distortion rule can satisfy any nontrivial guarantee.
For arbitrary randomized rules, we prove universal lower bounds: every finite-distortion rule has $\Omega(1/n)$ worst-case sensitivity (and $\Omega(1/n)$ average sensitivity as well), and when $\varepsilon_{\mathsf{DP}}=O(1/n)$, approximate differential privacy requires $\delta=\Omega(1/n)$.  Thus, the one-voter scale $1/n$ is unavoidable, and our upper bounds match this scale up to logarithmic factors; the necessity of the $\log m$ factor remains open (see \Cref{app:limitations}).

\paragraph{Baseline comparisons.}
For comparison, we analyze sensitivity and differential privacy for two existing randomized rules: Random Dictatorship, which has distortion $3-2/n$~\citep{AnshelevichPostl2017}, and the randomized rule of \citet{Kempe2020}, which has distortion $3-2/m$.  These distortion improvements over $3$ vanish as $n$ or $m$ grows.  Both rules satisfy an $O(1/n)$ bound for the weaker average sensitivity notion systematically studied by \citet{VarmaYoshida}, but their worst-case sensitivity is $\Omega(1)$.
For privacy, those rules are $(0,O(1/n))$-differentially private.  See \Cref{app:baselines} for details.  These comparisons motivate the question addressed by our main results: whether an absolute improvement over $3$ can coexist with worst-case sensitivity and approximate differential privacy close to the one-voter scale of $1/n$.
\subsection{Technical Overview}
Our construction starts from the bounded-randomness theorem of \citet{Cai2026}, used here in the list form stated in \Cref{prop:br-uniform-list}: there is an absolute constant $K$ such that every preference profile admits a $K$-list, an ordered list of $K$ candidates with repetitions allowed, whose induced uniform lottery has distortion strictly below $3$.  Through the existing biased-metric ratio characterization (\Cref{prop:biased-ratio}), this is equivalent to the existence of a $K$-list $L$ whose induced uniform lottery $p_L$ has normalized excess-cost value $\lambda_\sigma(L)\le 1-\gamma$ for an absolute constant $\gamma>0$.  For a fixed nonnegative vector~$x$ inducing a biased metric, the quantity has the form
\[
   \lambda_\sigma(L;x)
   =
   \frac{L_\sigma(p_L;x)}{R_\sigma(x)},
\]
and $\lambda_\sigma(L)$ is the supremum over $x$.  Here $L_\sigma(p_L;x)$ represents the excess-cost contribution of the induced uniform lottery, and $R_\sigma(x)$ is the normalization term determined by the voters' rankings $\sigma$.  The argument of \citet{Cai2026} gives a low-ratio list pointwise for each profile, but it does not by itself provide a rule with controlled sensitivity that selects one such list for each profile.  Rather than selecting a single low-ratio list at a profile, our rule randomizes over all $K$-lists and assigns each list $L$ a Gibbs weight $\exp(-\eta \lambda_\sigma(L))$, as formalized in \Cref{sec:gibbs}.  The inverse temperature $\eta$ is chosen on the order of $\log m^K \simeq K\log m$, so that the existence of a list $L$ with $\lambda_\sigma(L)\le 1-\gamma$ translates to a distortion guarantee below $3$ for the resulting lottery induced by the Gibbs distribution.

For the sensitivity and privacy analyses, the main task is to control how the ratio value $\lambda_\sigma(L)$ changes when one voter is deleted or one voter's ranking is replaced.  The analysis proceeds by focusing on sublevels of this ratio value; the contribution of lists with high ratios is controlled by Gibbs tail bounds.  Using the structure of the biased-metric ratio, we obtain a coverage inequality of the form $L_\sigma(p_L;x)+p_L(c)R_\sigma(x)\ge p_L(c)/2$ for some candidate $c$ in the support of $p_L$.  Since $p_L$ is obtained by drawing uniformly from the $K$ entries of $L$, the relevant mass $p_L(c)$ is at least $1/K$.  On a sublevel where $\lambda_\sigma(L;x)\le D$, combining the coverage inequality with the ratio bound implies a lower bound on the denominator $R_\sigma(x)$; see \cref{lem:denominator-lower-bound} for details.  This denominator bound is the crux of our stability analysis.
With it in hand, we show that deleting one voter or replacing one voter's ranking perturbs the $\lambda_\sigma(L)$ value by only $O(1/n)$ on these low-ratio sublevels, as shown using \Cref{prop:local-profile-ratio-change}.  This approach may also be useful for studying the biased-metric framework beyond the sensitivity and privacy applications considered here.

The remaining analysis uses this perturbation bound for the Gibbs rule.  For sensitivity, \Cref{prop:local-profile-ratio-change} gives the bound under a fixed voter deletion needed to control the change of the Gibbs distribution over lists; a transport argument then converts this list-level stability into the sensitivity bound in \Cref{thm:gibbs-over-lists}.  For differential privacy, the same proposition yields the bound under replacement of one voter's ranking on the low-ratio sublevels; combining this bound with the Gibbs tail bounds outside those sublevels gives \Cref{thm:gibbs-privacy}.
\section{Preliminaries}\label{sec:prelim}
Let $\Cands$ be the set of candidates with $|\Cands|=m$ and $\Voters$ the set of voters with $|\Voters|=n$.  Preferences are strict linear orders; ties in an underlying metric, if any, are broken arbitrarily before the profile is given to the rule.  We write $a\succ_v b$ if voter $v$ ranks $a$ above $b$.  An election instance is a tuple $\mathcal E=(\Voters,\Cands,\sigma)$, where $\sigma=(\sigma_v)_{v\in\Voters}$ is the reported preference profile.  The hidden metric information is represented by a possibly non-injective embedding $\phi\colon\Voters\cup\Cands\to(\mathcal M,d_{\mathcal M})$ into a metric space.  We write $d(x,y)=d_{\mathcal M}(\phi(x),\phi(y))$, so $d$ is a pseudometric on the labeled voters and candidates.  We write $d\consistent\mathcal E$ if this pseudometric is consistent with the election instance, meaning that $a\succ_v b$ implies $d(v,a)\le d(v,b)$ for every voter $v$ and candidates $a,b$.
When $\Voters$ and $\Cands$ are fixed, we identify the election instance with its profile and also write $d\consistent\sigma$.
For a candidate $c$ and a lottery $p\in\triangle(\Cands)$, define the average social costs by
\[
   \SC_{d,\sigma}(c)=\frac1n\sum_{v\in\Voters}d(v,c),
   \quad
   \SC_{d,\sigma}(p)=\sum_{c\in\Cands}p(c)\SC_{d,\sigma}(c),
\]
and let $\OPT_{d,\sigma}=\min_{c\in\Cands}\SC_{d,\sigma}(c)$.
When the profile is clear, we write simply $\SC_d$ and $\OPT_d$.

A randomized voting rule $F$ maps each election instance $\mathcal E$ to a lottery over its candidates.  For $\mathcal E=(\Voters,\Cands,\sigma)$, we write this lottery as $F(\sigma)\in\triangle(\Cands)$.  Its distortion is defined as
\[
   \Dist(F)
   =
   \sup_{\mathcal E=(\Voters,\Cands,\sigma)}\sup_{d:d\consistent\mathcal E}
   \frac{\SC_{d,\sigma}(F(\sigma))}{\OPT_{d,\sigma}},
\]
where a ratio with denominator $0$ is interpreted in the extended sense: it is $0$ if the numerator is $0$ and $+\infty$ otherwise.  This convention keeps zero-optimum instances in the distortion definition and forces every finite-distortion rule to have zero rule cost whenever $\OPT_{d,\sigma}=0$.  The biased-metric ratio formulation in \cref{sec:biased-ratio} uses the same extended-ratio convention.
\subsection{Sensitivity}\label{subsec:sensitivity}
We use metric-normalized worst-case sensitivity under one-voter deletion. For two lotteries $p,q\in\triangle(\Cands)$, let $\W_d(p,q)$ denote their 1-Wasserstein distance with respect to the restriction of $d$ to the candidates, defined by
\[
   \W_d(p,q)
   =
   \min_{\pi\in\Gamma(p,q)}\sum_{a,b\in\Cands}\pi(a,b)d(a,b),
\]
where $\Gamma(p,q)$ is the set of couplings of $p$ and $q$.
If $\sigma^{-v}$ is the profile obtained by deleting voter $v$, then for fixed numbers $m$ of candidates and $n\ge2$ of voters, we define the size-fixed worst-case sensitivity by
\[
   \Sens_{m,n}(F)
   =\!\!\!\!\!\!\!
   \sup_{\substack{\mathcal E=(\Voters,\Cands,\sigma),\ |\Cands|=m,\\  |\Voters|=n,\ d\consistent\mathcal E,\ \OPT_{d,\sigma}>0}}
   \!\!\!
   \max_{v\in\Voters}
   \frac{\W_d(F(\sigma),F(\sigma^{-v}))}{\OPT_{d,\sigma}}.
\]
The supremum ranges over all election instances of the given size and all consistent pseudometrics.  The assumption $n\ge2$ ensures that each profile after deletion has at least one voter.  We use the same pseudometric $d$, restricted to $\Cands$, to measure the Wasserstein distance between $F(\sigma)$ and $F(\sigma^{-v})$; only the profile input to the rule changes.
The normalization by $\OPT_{d,\sigma}$ makes the Wasserstein term scale invariant, since both the numerator and denominator scale linearly with $d$.  The condition $\OPT_{d,\sigma}>0$ excludes degenerate normalized terms of the form $0/0$.  For finite-distortion rules, the excluded terms have zero numerator and equal $0$ under the extended-ratio convention; see \Cref{app:normalization-conventions}.  We use the pre-deletion optimum because $\OPT_{d,\sigma^{-v}}$ can be zero even when $\OPT_{d,\sigma}>0$.
Average sensitivity $\AS_{m,n}(F)$, defined formally in \Cref{app:baselines}, replaces the maximum by the average over deleted voters, and therefore satisfies $\AS_{m,n}(F)\le\Sens_{m,n}(F)$.
\subsection{Differential Privacy}\label{subsec:privacy}
We say two profiles $\sigma$ and $\tau$ on the same voter and candidate sets are \emph{replacement-neighboring} if they differ in the reported ranking of exactly one voter.  A randomized voting rule $F$ determines a distribution $F(\sigma)\in\triangle(\Cands)$, and the associated released-winner mechanism draws a single candidate $Y_\sigma\sim F(\sigma)$.  For $\varepsilon,\delta\ge0$, the released-winner mechanism is $(\varepsilon,\delta)$-differentially private under replacement adjacency if, for every pair of replacement-neighboring profiles $\sigma,\tau$ and every event $S\subseteq\Cands$, the released winner satisfies
\[
   \Prob[Y_\sigma\in S]
   \le
   \e^\varepsilon\Prob[Y_\tau\in S]+\delta.
\]
The case of $\delta=0$ is pure differential privacy.  The definition follows the standard formulation \citep{DworkMcSherryNissimSmith2006,DworkRoth2014} applied to the released candidate rather than to the full lottery $F(\sigma)$.
Note that the mechanism receives the full reported profile, and only the sampled winner is released.\footnote{
   This is distinct from ballot secrecy, which aims to prevent linking an identified voter to the content of that voter's ballot, including by election officials or other observers.
}
\subsection{\texorpdfstring{$K$-Lists}{K-Lists} and Induced Lotteries}
Throughout this paper, for $K\in\mathbb N_{\ge1}$, a $K$-list is an ordered $K$-tuple $L=(c_1,\ldots,c_K)\in\Cands^K$.  All sums over such lists count ordered tuples in $\Cands^K$.  The list induces the uniform lottery
\[
   p_L(c)=\frac1K\abs*{\Set*{j}{c_j=c}}.
\]
Let $\Lists_K=\Cands^K$ denote the set of all $K$-lists, so $|\Lists_K|=m^K$.  Different ordered tuples may induce the same lottery, but they remain distinct list states; equivalently, we count lists rather than induced lotteries.

For a fixed profile $\sigma$ and lottery $p$, define the worst-case distortion of $p$ on profile $\sigma$ by
\[
   \rho_\sigma(p)
   =
   \sup_{d:d\consistent\sigma}
   \frac{\SC_{d,\sigma}(p)}{\OPT_{d,\sigma}}.
\]
For a list $L$, write $\rho_\sigma(L)=\rho_\sigma(p_L)$.
\section{Biased-Metric Framework for Lists}\label{sec:biased-ratio}
Biased metrics were introduced by \citet{CharikarRamakrishnan2022}; the ratio characterization used here is due to \citet{CRWW2024} and is also used by \citet{Cai2026} in their bounded-randomness result.  We explain this notion, state the bounded-randomness result used by our construction, and provide a support-normalized reduction for list lotteries.

For a fixed profile $\sigma$ and lottery $p$, the biased-metric characterization represents the worst-case distortion as $\rho_\sigma(p)=1+2\lambda_\sigma(p)$, where $\lambda_\sigma(p)$ is the supremum of the fixed-vector ratios defined below.  For lists, we write $\lambda_\sigma(L)=\lambda_\sigma(p_L)$.

To define this ratio, consider the normalized set of vectors
\[
   \X=\Set*{x\in[0,1]^{\Cands}}{\min_c x_c=0,\ \max_c x_c=1}.
\]
In the biased-metric interpretation, $x_j$ represents the distance from an optimal candidate to candidate $j$, after scaling.  The normalization loses no generality because the ratios are scale invariant; it fixes the lowest candidate level at $0$ and the highest at $1$.  For $x\in\X$ and $t\ge0$, define the lower level set
\[
   I_t(x)=\Set*{j\in\Cands}{x_j\le t}.
\]
Thus $I_t(x)$ contains the candidates within distance level $t$ of the optimum.

We now define the quantities that depend on the reported profile.  For a subset $A\subseteq\Cands$ and candidate $j$, let
\[
   s_\sigma(A\succ j)
   =
   \frac1n\abs*{\Set*{v\in\Voters}{\forall a\in A,\ a\succ_v j}},
\]
where $a\succ_v j$ denotes strict preference in voter $v$'s reported linear order.  Thus $s_\sigma(A\succ j)$ is the fraction of voters who rank every candidate in $A$ above $j$.  If $A=\emptyset$, the condition is vacuous and $s_\sigma(\emptyset\succ j)=1$.  If $j\in A$, then $s_\sigma(A\succ j)=0$, since no candidate is strictly above itself.

For a lottery $p\in\triangle(\Cands)$, set
\[
   \ell_{\sigma,p,x}(t)
   =
   \sum_{j\notin I_t(x)}s_\sigma(I_t(x)\succ j)p(j).
\]
At level $t$, this quantity looks only at lottery mass on candidates still outside the lower level set and asks what fraction of voters rank the entire lower level set above such a candidate.  In this sense $\ell_{\sigma,p,x}(t)$ is the lottery-dependent threshold term used in the ratio numerator below.

For each voter $v$, define
\[
   g_v(x)=\max_{a\succ_v b}(x_a\!-\!x_b)_+,
   \quad
   r_{\sigma,x}(t)
   =\frac1n\abs*{\Set*{v\in\Voters\!}{\!g_v(x)>t}}.
\]
where $y_+\coloneqq\max\{y,0\}$.  The quantity $g_v(x)$ is the largest positive increase in $x$ along a preference comparison of voter $v$: if voter $v$ ranks $a$ above $b$ while $x_a>x_b$, then $a$ is farther from the reference optimum than $b$ is by $x_a-x_b$.  Thus $r_{\sigma,x}(t)$ is the fraction of voters for which this largest gap exceeds $t$.

Finally define the integrated lottery-dependent term by
\[
   L_\sigma(p;x)=\int_0^\infty \ell_{\sigma,p,x}(t)\,\diff t.
\]
Also define the voter-dependent normalization by
\[
      R_\sigma(x)
   =\int_0^\infty r_{\sigma,x}(t)\,\diff t
   =\frac1n\sum_v\int_0^\infty \mathds{1}(g_v(x)>t)\,\diff t
   =\frac1n\sum_v g_v(x),
\]
where $\mathds{1}(\cdot)$ takes $1$ if the argument is true and $0$ otherwise.
Since $x\in[0,1]^{\Cands}$, all integrands vanish for $t\ge1$.
For a fixed vector $x$, the ratio $\lambda_\sigma(p;x) = L_\sigma(p;x)/R_\sigma(x)$ compares the integrated lottery-dependent term $L_\sigma(p;x)$ with the voter-dependent normalization $R_\sigma(x)$.  The biased-metric characterization says that the supremum of this fixed-vector ratio over $x$ is exactly the excess-distortion parameter $\lambda_\sigma(p)$.
\begin{proposition}[Biased-Metric Ratio Characterization; {\normalfont\citealt[Theorem~2.8]{Cai2026}, after \citealt{CRWW2024}}]\label{prop:biased-ratio}
Assume $|\Cands|\ge2$.  For every profile $\sigma$ and lottery $p$, we have
\[
   \rho_\sigma(p)=1+2\lambda_\sigma(p),
\]
where $\lambda_\sigma(p)=\sup_{x\in\X}\lambda_\sigma(p;x)$ and, for each $x\in\X$,
\[
   \lambda_\sigma(p;x)=
   \begin{cases}
   L_\sigma(p;x)/R_\sigma(x),& R_\sigma(x)>0,\\[1mm]
   0,& R_\sigma(x)=L_\sigma(p;x)=0,\\[1mm]
   +\infty,& R_\sigma(x)=0<L_\sigma(p;x).
   \end{cases}
\]
\end{proposition}
When $m=1$, we use the convention $\lambda_\sigma(p)=0$; every rule outputs the unique candidate, so $\rho_\sigma(p)=1$.

This is the biased-metric characterization used in recent improvements over the $3$-distortion barrier for randomized rules; in this notation it follows from \citet[Theorem~2.8]{Cai2026}, with attribution there to \citet{CRWW2024}.  We use the same notation for lists by applying the characterization to the induced lottery $p_L$.  \Cref{app:biased-ratio-details} gives the translation from the existing statement and the convention details.

With this notation in place, the bounded-randomness theorem \citep{Cai2026} gives the following low-distortion list, which immediately implies a bound on $\lambda_\sigma(L)$.
\begin{proposition}[List Consequence of Bounded Randomness; {\normalfont\citealt[Theorem~7.6, Corollary~7.7, Remark~7.8]{Cai2026}}]\label{prop:br-uniform-list}
There are absolute constants $K_{\mathsf{BR}}\in\mathbb N_{\ge1}$ and $\varepsilon_{\mathsf{BR}}>0$ such that every preference profile $\sigma$ admits a $K_{\mathsf{BR}}$-list $L$ with
\[
   \rho_\sigma(L)
   \le
   3-\varepsilon_{\mathsf{BR}}.
\]
Using the identity $\rho_\sigma(L)=1+2\lambda_\sigma(L)$ from \Cref{prop:biased-ratio}, this implies that, with $\gamma=\varepsilon_{\mathsf{BR}}/2$, we have
\[
   \lambda_\sigma(L)\le 1-\gamma.
\]
\end{proposition}
\begin{remark}[Use of the Result of \citealt{Cai2026}]
\Cref{prop:br-uniform-list} is the only property of the construction of \citet{Cai2026} used in our analysis.
\Cref{app:br-consequence} gives the details of how the result of \citet{Cai2026} translates to the list form used here.  In the appendix, we also give a conservative calculation showing that the Gibbs rule constructed in \Cref{sec:gibbs} has distortion less than $3-5.9\times10^{-9}$.
\end{remark}
For our stability analysis below, it is useful to restrict the fixed-vector supremum for a list to vectors whose maximum is attained on the support of the list.  For a list $L$, define the set of support-normalized vectors
\[
   \Xsn(L)=\Set[\Big]{x\in\X}{\max_{c\in\Supp(p_L)}x_c=1}.
\]
Here $\Supp(p)=\Set*{c\in\Cands}{p(c)>0}$ is the support of a lottery $p$.  When $|\Cands|\ge2$, the supremum for a list can always be taken over this smaller set.
\begin{restatable}[Support-Normalized Reduction]{lemma}{SupportNormalizedReduction}\label{lem:support-normalized-reduction}
Assume $|\Cands|\ge2$.
For every profile $\sigma$ and every list $L$, it holds that
\[
   \lambda_\sigma(L)
   =
   \sup\Set*{\lambda_\sigma(L;x)}{x\in\Xsn(L)}.
\]
\end{restatable}
See \Cref{app:support-normalized-proof} for the proof.
Intuitively, it is a clipping-and-rescaling argument.  For any vector $x$, let $h$ be the largest $x$-value among candidates in the support of the list.  If $h>0$, replacing $x$ coordinatewise by $\min\{x,h\}/h$ makes the vector support-normalized, preserves the lower level sets relevant to the support under the change of variables from $t$ to $u=ht$, and cannot decrease the ratio.
If $h=0$, the numerator is zero as no positive-mass candidate lies outside any lower level set.
\section{Stability of Bounded Biased-Metric Ratios}\label{sec:list-stability}
This section presents a stability argument that plays a key role in the sensitivity and differential-privacy analysis.
\subsection{Denominator Lower Bound on Low Sublevels}
We first control the change of the ratio under voter deletion or replacement.  To this end, we give a denominator lower bound on low sublevels of $\lambda_\sigma(L;x)$.  Specifically, if $x\in\Xsn(L)$ satisfies $\lambda_\sigma(L;x)\le D$, $R_\sigma(x)$ is bounded below as follows.
\begin{lemma}[Denominator Lower Bound]\label{lem:denominator-lower-bound}
Let $L$ be a $K$-list and let $x\in\Xsn(L)$ and $D\ge0$.
If $\lambda_\sigma(L;x)\le D$, then the denominator satisfies
\[
   R_\sigma(x)
   \ge
   \underline R_{K,D}
   \coloneqq
   \frac{1}{2(KD+1)}.
\]
In particular, if $\lambda_\sigma(L)\le D$, then the same lower bound holds for every $x\in\Xsn(L)$.
\end{lemma}
\begin{proof}
From $x\in\Xsn(L)$, there is a candidate $c\in\Supp(p_L)$ with $x_c=1$.  Fix one such candidate and put $w_c=p_L(c)$.  Since $p_L$ is obtained by drawing uniformly from the $K$ entries of $L$, every support candidate has mass at least $1/K$, hence $w_c\ge1/K$.

For $t\in[0,1/2)$, let $q(t)=s_\sigma(I_t(x)\succ c)$.  Since $x_c=1$, the candidate $c$ is not in $I_t(x)$ for $t<1$.  We claim that
\begin{equation}\label{eq:q-r-cover}
   q(t)+r_{\sigma,x}(t)\ge1,
   \qquad 0\le t<1/2.
\end{equation}
To prove \eqref{eq:q-r-cover}, it suffices to show that every voter counted by $1-q(t)$ is also counted by $r_{\sigma,x}(t)$.  If $I_t(x)=\emptyset$, then we have $s_\sigma(\emptyset\succ c)=1$ by vacuity, so $q(t)=1$ and the claim follows; if no voter is counted by $1-q(t)$, then $1-q(t)=0$ holds and the claim follows.  Otherwise, take any voter $v$ who does not rank every candidate in $I_t(x)$ strictly above $c$.  Then there is some $b\in I_t(x)$ that voter $v$ does not rank strictly above~$c$.  Since preferences are strict complete orders and $c\notin I_t(x)$ ensures $c\neq b$, we have $c\succ_v b$, and hence it holds that
\[
   g_v(x)\ge (x_c-x_b)_+\ge 1-t>t
\]
since $x_c=1$, $x_b\le t$, and $t<1/2$.
Thus, from the definition of $r_{\sigma,x}(t)$, we conclude that every voter counted by $1-q(t)$ is counted by $r_{\sigma,x}(t)$, i.e., $1-q(t)\le r_{\sigma,x}(t)$, thus proving \eqref{eq:q-r-cover}.

We next derive the coverage inequality that will be used to lower bound $R_\sigma(x)$.  Since $c\notin I_t(x)$ for $t<1$, the summand corresponding to $c$ appears in the sum defining $\ell_{\sigma,p_L,x}(t)$, so
\[
   \ell_{\sigma,p_L,x}(t)
   =\sum_{a\notin I_t(x)}p_L(a)s_\sigma(I_t(x)\succ a)
   \ge w_c q(t).
\]
Combining this with \eqref{eq:q-r-cover} gives
\[
   \ell_{\sigma,p_L,x}(t)+w_c r_{\sigma,x}(t)
   \ge w_c
   \qquad 0\le t<1/2.
\]
Since \eqref{eq:q-r-cover} holds on $[0,1/2)$, integrating over this interval and using nonnegativity outside it gives
\[
   L_\sigma(p_L;x)+w_c R_\sigma(x)
   \ge \frac{w_c}2.
\]
If $R_\sigma(x)=0$, this inequality forces $L_\sigma(p_L;x)\ge w_c/2>0$, hence $\lambda_\sigma(L;x)=+\infty$, contrary to the assumption.  Therefore, we have $R_\sigma(x)>0$.  Since $\lambda_\sigma(L;x)\le D$, we have $L_\sigma(p_L;x)\le D R_\sigma(x)$, and the preceding inequality gives
\[
   (D+w_c)R_\sigma(x)\ge \frac{w_c}2.
\]
Using $w_c\ge1/K$ and monotonicity of $w\mapsto w/(D+w)$ on $[0,\infty)$, we obtain
\[
   R_\sigma(x)
   \ge \frac{w_c}{2(D+w_c)}
   \ge\frac{1/K}{2(D+1/K)}
   =\underline R_{K,D}.
\]
The final sentence of the lemma follows since, by \Cref{lem:support-normalized-reduction}, $\lambda_\sigma(L)$ is the supremum of $\lambda_\sigma(L;x)$ over $x\in\Xsn(L)$; hence $\lambda_\sigma(L)\le D$ implies $\lambda_\sigma(L;x)\le D$ for every such $x$.
\end{proof}
\subsection{Bounding the Effect of Local Profile Changes}\label{subsec:local-profile-ratio-change}
Let $L$ be a $K$-list and let $\sigma,\sigma'$ be profiles with nonempty voter sets on the same candidate set.  In the following analysis, we use only the following additive perturbation condition: for some $\Delta>0$, every $x\in\Xsn(L)$ satisfies
\begin{equation}\label{eq:local-profile-average-drift}
   |L_{\sigma'}(p_L;x)-L_\sigma(p_L;x)|\le \Delta,
   \quad
   |R_{\sigma'}(x)-R_\sigma(x)|\le \Delta.
\end{equation}
This condition holds with $\Delta=1/n$ for one-voter changes---both deletion and replacement---where $n$ is the number of voters in $\sigma$.  Indeed, suppose $\sigma'$ is obtained from $\sigma$ either by deleting one voter, with $n\ge2$, or by replacing one voter's ranking.  For any fixed $x\in\X$ and $t\ge0$, both $\ell_{\sigma,p_L,x}(t)$ and $r_{\sigma,x}(t)$ are averages over voters of quantities in $[0,1]$.  If one term $a_v$ is deleted from an average of $n$ numbers in $[0,1]$, we have
$
   \left|\frac1n\sum_{i=1}^n a_i-\frac1{n-1}\sum_{i\ne v}a_i\right|
   =
   \frac{|a_v-\bar a_{-v}|}{n}
   \le\frac1n
$,
where $\bar a_{-v}$ is the average over the remaining terms.  Similarly, replacing one term in an average of $n$ terms in $[0,1]$ also changes the average by at most $1/n$.  Since those functions vanish for $t\ge1$, integrating the pointwise bound gives \eqref{eq:local-profile-average-drift} with $\Delta=1/n$.
The denominator lower bound in \cref{lem:denominator-lower-bound} converts the condition \eqref{eq:local-profile-average-drift} into the following bound on $|\lambda_{\sigma'}(L)-\lambda_\sigma(L)|$.
\begin{restatable}[Ratio Change Bound under Local Profile Changes]{proposition}{LocalProfileRatioChangeBound}\label{prop:local-profile-ratio-change}
Let $L$ be a $K$-list, and let $\sigma,\sigma'$ be profiles with nonempty voter sets on the same candidate set.  Let $\Delta>0$ be such that \eqref{eq:local-profile-average-drift} holds for every $x\in\Xsn(L)$, and define
\[
   A_K(t)\coloneqq 4(t+1)(Kt+1).
\]
If $\lambda_\sigma(L)<\infty$ and $\Delta\le 1/(4(K\lambda_\sigma(L)+1))$, then we have $\lambda_{\sigma'}(L)<\infty$ and
\[
   |\lambda_{\sigma'}(L)-\lambda_\sigma(L)|
   \le
   A_K(\lambda_\sigma(L))\Delta.
\]
The symmetric statement is also true: if $\lambda_{\sigma'}(L)<\infty$ and $\Delta\le 1/(4(K\lambda_{\sigma'}(L)+1))$, then we have $\lambda_\sigma(L)<\infty$ and
\[
   |\lambda_{\sigma'}(L)-\lambda_\sigma(L)|
   \le
   A_K(\lambda_{\sigma'}(L))\Delta.
\]
In particular, if $\sigma$ has $n$ voters and $\sigma'$ is obtained from $\sigma$ either by deleting one voter, with $n\ge2$, or by replacing one voter's ranking, then \eqref{eq:local-profile-average-drift} holds with $\Delta=1/n$; hence the first implication applies whenever $1/n \le 1/(4(K\lambda_\sigma(L)+1))$, and the symmetric implication applies under the analogous condition with $\lambda_{\sigma'}(L)$.
\end{restatable}
\begin{proof}[Proof sketch]
By \Cref{lem:support-normalized-reduction}, it suffices to compare the fixed-vector ratios on $\Xsn(L)$.  Let $\Lambda=\lambda_\sigma(L)$ and fix $x\in\Xsn(L)$.  The denominator bound in \Cref{lem:denominator-lower-bound} gives $R_\sigma(x)\ge \underline R_{K,\Lambda}$, while \eqref{eq:local-profile-average-drift} changes both the integrated numerator and denominator by at most $\Delta$.  The assumption $\Delta\le1/(4(K\Lambda+1))$ is the same as $\Delta\le \underline R_{K,\Lambda}/2$, so this additive perturbation is at most half of the denominator lower bound.  Expanding the two quotients with denominators $R_\sigma(x)$ and $R_{\sigma'}(x)$ gives the stated $A_K(\Lambda)\Delta$ bound uniformly over $x$.  Taking suprema over $x\in\Xsn(L)$ gives the stated bound on $|\lambda_{\sigma'}(L)-\lambda_\sigma(L)|$.  The reverse implication follows from the same calculation with the roles of $\sigma$ and $\sigma'$ interchanged, using the same additive perturbation bound \eqref{eq:local-profile-average-drift}.

The one-voter instantiations are the cases where \eqref{eq:local-profile-average-drift} holds with $\Delta=1/n$.  See \Cref{app:local-profile-ratio-change-proof} for the full proof.
\end{proof}
\section{Gibbs Distributions over \texorpdfstring{$K$-Lists}{K-Lists}}\label{sec:gibbs}
We construct the voting rule by using the Gibbs distribution over all $K$-lists.  This is an exponential-mechanism-type construction~\citep{McSherryTalwar2007} with score $-\lambda_\sigma(L)$.  Note that the standard global-sensitivity analysis is not directly applicable because $\lambda_\sigma(L)$ need not have a useful global sensitivity bound over all $K$-lists; the analysis below uses local drift on low-energy sublevels, established in \Cref{prop:local-profile-ratio-change}, together with Gibbs tail bounds.

Fix a list size $K\in\mathbb N_{\ge1}$.  For a profile $\sigma$ and an inverse temperature $\eta>0$, define the Gibbs distribution on all $K$-lists by
\[
   \Pi_{\sigma}^{(\eta)}(L)
   =\frac{\exp(-\eta\lambda_\sigma(L))}{\sum_{M\in\Lists_K}\exp(-\eta\lambda_\sigma(M))},
\]
where $\exp(-\eta\cdot\infty)=0$.\footnote{We sum over ordered tuples in $\Lists_K$; thus different orderings of the same multiset are distinct Gibbs states, even though they induce the same candidate lottery.}  In what follows, this definition is used only when at least one $K$-list has finite energy, so the denominator is strictly positive.  The induced candidate lottery is
\[
   Q_{\sigma}^{(\eta)}=\E_{L\sim\Pi_{\sigma}^{(\eta)}}[p_L].
\]
This is the candidate distribution determined by the rule at profile $\sigma$ and temperature $\eta$; equivalently, one may first draw $L\sim\Pi_{\sigma}^{(\eta)}$ and then draw the winner uniformly from $L$.
We refer to $\lambda_\sigma(L)$ as the Gibbs energy of list $L$ on profile $\sigma$.

Fixing constants $K\in\mathbb N_{\ge1}$ and $\gamma > 0$, we use the following one-parameter temperature family: for $u\ge1$, define
\[
   \eta_u\coloneqq\frac{8}{\gamma}\bigl(K\log m+u+1\bigr).
\]
Below we write $Q^{(u)}_\sigma=Q_\sigma^{(\eta_u)}$.  Since $|\Lists_K|=m^K$, the Gibbs variational bound (see \Cref{lem:gibbs-variational}) shows that if some list has energy at most $1-\gamma$, then taking $\eta$ on the scale of $K\log m$ ensures that the Gibbs expectation of $\lambda_\sigma(L)$ is less than $1$, which implies a distortion bound below $3$ through \cref{prop:biased-ratio}.  The additional parameter $u$ is used later in the tail bounds:
the sensitivity theorem takes $u=1$; the approximate-differential-privacy theorem takes $u=\max\{1,\log(1/\delta)\}$.

We also define the associated released-winner mechanism as the following two-step sampling procedure:
\[
   L_\sigma^{(u)}\sim\Pi_\sigma^{(\eta_u)},
   \qquad
   Y_\sigma^{(u)}\mid L_\sigma^{(u)}\sim p_{L_\sigma^{(u)}}.
\]
It releases the candidate $Y_\sigma^{(u)}$, whose law is $Q_\sigma^{(u)}$; the privacy guarantee shown later does not release the lottery $Q_\sigma^{(u)}$ itself.
\begin{remark}[On basic social-choice properties]\label{rem:gibbs-basic-properties}
The Gibbs rule is anonymous and neutral, as relabeling voters or candidates only relabels the energies and list states.  It is also unanimous: if all voters rank one candidate first, the finite distortion implies that the Gibbs distribution is supported on lists inducing the unanimous candidate lottery.  See \Cref{app:gibbs-axioms} for details.  However, we make no claim of Condorcet consistency, participation, or strategyproofness.
\end{remark}
\section{Sensitivity}\label{sec:sensitivity}
We give the worst-case sensitivity bound for the Gibbs rule.
The next lemma compares two Gibbs distributions whose energies satisfy the local drift condition as in \Cref{prop:local-profile-ratio-change}.
\begin{restatable}[Weighted Gibbs Stability under Local Energy Change]{lemma}{EnergyDependentWeightedGibbsStability}\label{lem:energy-dependent-gibbs}
Fix $K\in\mathbb N_{\ge1}$ and $B\ge0$, and define $A_K$ as in \Cref{prop:local-profile-ratio-change}.
There exist constants $\nu_{K,B},C_{K,B}\in(0,\infty)$, depending only on $K$ and $B$, with the following property.

Let $\Omega$ be a finite set, $E,E'\colon\Omega\to[0,\infty]$ be energies on $\Omega$, and $\pi,\pi'$ be their Gibbs distributions at a common inverse temperature $\eta\ge1$.  Assume $\min_{\omega\in\Omega}E(\omega)\le B$, $\min_{\omega\in\Omega}E'(\omega)\le B$, and $\eta\ge\log|\Omega|$.  Let $\Delta\in(0,1]$ satisfy $\eta\Delta\le \nu_{K,B}$.  Suppose that the energies satisfy the following local drift condition: whenever $E(\omega)<\infty$ and $\Delta\le1/(4(KE(\omega)+1))$, the value $E'(\omega)$ is finite and satisfies
\[
   |E'(\omega)-E(\omega)|
   \le
   A_K(E(\omega))\Delta,
\]
and the same implication holds with $E$ and $E'$ interchanged.  Under these assumptions, it holds that
\begin{equation}
\mainbodyequationlabel{eq:energy-weighted-positive-part-bound}
   \sum_{\smash{\substack{\omega:E(\omega)<\infty}}}
   (\pi(\omega)-\pi'(\omega))_+(1+E(\omega))
   +
   \sum_{\substack{\omega:E'(\omega)<\infty}}
   (\pi'(\omega)-\pi(\omega))_+(1+E'(\omega))
   \le 
   C_{K,B}\eta\Delta.
\end{equation}
\end{restatable}
\begin{proof}[Proof sketch]
   Put $s=\eta\Delta$ and choose $M=(16C_Ks)^{-1/2}$, where $C_K$ is a constant such that $A_K(t)\le C_K(1+t)^2$.  Choose $\nu_{K,B}$ small enough so that, whenever $s\le\nu_{K,B}$, inequalities $\Delta\le1/(4(KE(\omega)+1))$ and $\Delta\le1/(4(KE'(\omega)+1))$ hold on the sublevels $E(\omega)\le M$ and $E'(\omega)\le M$, and $\max\set*{sA_K(E(\omega)),\ sA_K(E'(\omega))}\le1/4$ holds there.  On these sublevels, the local drift condition ensures that replacing $E$ by $E'$ changes each unnormalized Gibbs weight by only a controlled multiplicative amount. Moreover, for each energy $G\in\{E,E'\}$ and its Gibbs law $\mu$, the Gibbs tail bounds give $\sum_{\omega:M<G(\omega)<\infty}\mu(\omega)(1+G(\omega))=O_{K,B}(s)$.
Similarly, comparing the partition sums gives $|\log(Z'/Z)|=O_{K,B}(s)$.
These partition-function and tail bounds together imply \eqref{eq:energy-weighted-positive-part-bound}.
   See \Cref{app:gibbs-over-lists-proof} for the full proof.
\end{proof}
In \eqref{eq:energy-weighted-positive-part-bound}, the restrictions to finite-energy states only avoid undefined factors:
if $E(\omega)=\infty$, then $\pi(\omega)=0$ holds, and thus the corresponding positive part in the first sum is zero, and similarly for $E'$ and the second sum.

\Cref{lem:energy-dependent-gibbs} converts local stability of the list energies into the bound \eqref{eq:energy-weighted-positive-part-bound} on the energy-weighted discrepancy between the two Gibbs distributions.  This gives the following worst-case sensitivity bound while keeping the distortion below $3$.
\begin{restatable}[Worst-Case Sensitivity Bound]{theorem}{AverageSensitivityUniformListTheorem}\label{thm:gibbs-over-lists}\label{thm:main}
Suppose there exist constants $K\in\mathbb N_{\ge1}$ and $\gamma\in(0,1)$ such that every profile $\sigma$ admits a $K$-list $L$ with $\lambda_\sigma(L) \le 1-\gamma$.
For $u\ge1$, let $Q^{(u)}$ be the Gibbs rule over all $K$-lists with inverse temperature $\eta_u\coloneqq(8/\gamma)(K\log m+u+1)$.  Its distortion satisfies
\[
   \Dist(Q^{(u)})
   \le
   3-\frac{7\gamma}{4}.
\]
Moreover, there is a constant $C_{K,\gamma}<\infty$, depending only on $K$ and $\gamma$, such that, for every $m\ge1$ and $n\ge2$, we have
\[
   \Sens_{m,n}(Q^{(u)})
   \le
   C_{K,\gamma}\frac{K\log m+u+1}{n}.
\]
The same bound holds for average sensitivity $\AS_{m,n}(Q^{(u)})\le\Sens_{m,n}(Q^{(u)})$.
In particular, by \Cref{prop:br-uniform-list}, there is a randomized ordinal voting rule $Q$, together with absolute constants $\varepsilon,C>0$, such that, for every $m\ge1$ and $n\ge2$,
\[
   \Dist(Q)\le 3-\varepsilon,\qquad
   \Sens_{m,n}(Q)\le C\frac{\log m+1}{n}.
\]
\end{restatable}
\begin{proof}[Proof sketch]
The distortion part first bounds the Gibbs expectation of the ratio value.  Applying the Gibbs variational inequality in \Cref{lem:gibbs-variational}, the assumed list of energy at most $1-\gamma$ and the bound $|\Lists_K|\le m^K$ give $\E[\lambda_\sigma(L)]\le1-7\gamma/8$.  With $Q_\sigma^{(u)}=\E[p_L]$, this gives $\rho_\sigma(Q_\sigma^{(u)})=\rho_\sigma(\E[p_L])\le\E[\rho_\sigma(p_L)]\le\E[1+2\lambda_\sigma(L)]\le3-7\gamma/4$, where the first inequality uses convexity of $\rho_\sigma$ in the lottery, since it is defined as a supremum of normalized social-cost functions that are linear in the lottery, and the second uses \Cref{prop:biased-ratio} for each list.

For sensitivity, compare a profile $\sigma$ with $\sigma'=\sigma^{-v}$ for a fixed voter $v$.  If $\OPT_{d,\sigma}=0$, $\W_d(Q^{(u)}(\sigma),Q^{(u)}(\sigma'))=0$ holds for finite-distortion rules.  Otherwise put $\Delta=1/n$ and $B=1-\gamma$, and set $E(L)=\lambda_\sigma(L)$ and $E'(L)=\lambda_{\sigma'}(L)$.  Suppose first that $\eta_u\Delta\le \nu_{K,B}$ holds, as required by \Cref{lem:energy-dependent-gibbs}.  In this case, \Cref{prop:local-profile-ratio-change} provides the symmetric local drift condition required by \Cref{lem:energy-dependent-gibbs}. Therefore, we obtain \eqref{eq:energy-weighted-positive-part-bound} with right-hand side $O(\eta_u/n)$ for the two Gibbs distributions over lists.  Transporting unmatched mass through an optimal candidate implies the normalized Wasserstein bound for this fixed deletion.  If instead $\eta_u\Delta>\nu_{K,B}$, the distortion bound and the resulting lower bound on $(K\log m+u+1)/n$ control the Wasserstein distance directly.  Since the deleted voter was arbitrary, this proves the worst-case sensitivity statement.  The final instantiation follows by taking $u=1$ and using \Cref{prop:br-uniform-list} with $K=K_{\mathsf{BR}}$ and $\gamma=\varepsilon_{\mathsf{BR}}/2$.  See \Cref{app:gibbs-over-lists-proof} for the full proof.
\end{proof}
\section{Differential Privacy}\label{sec:approx-privacy}
For differential privacy, neighboring profiles have the same number of voters and differ in one ranking.  The next lemma compares event probabilities for two Gibbs distributions by controlling likelihood ratios on a low-energy sublevel set.
\begin{restatable}[Event-Probability Bound for Gibbs Distributions]{lemma}{GibbsPrivacyComparisonLemma}\label{lem:gibbs-privacy-comparison}
Let $E,E'\colon\Omega\to[0,\infty]$ be energies on a finite set $\Omega$, and let $\pi,\pi'$ be their Gibbs distributions at inverse temperature $\eta\ge1$.  Fix $B\ge0$ and assume $\min_{\omega\in\Omega}E(\omega)\le B$ and $\min_{\omega\in\Omega}E'(\omega)\le B$.  Fix $D>B$ and $a,\Delta\ge0$, and let
\[
   \Omega_{\le D}
   =
   \Set*{\omega\in\Omega}{\min(E(\omega),E'(\omega))\le D}.
\]
Suppose that $|E(\omega)-E'(\omega)|\le a\Delta$ holds for all $\omega\in\Omega_{\le D}$, and put $\theta\coloneqq\eta a\Delta$ and $\kappa\coloneqq |\Omega|\e^{-\eta(D-B)}$.  Then, for every event $S\subseteq\Omega$, the probabilities satisfy
\[
   \pi(S)\le \e^{2\theta}\pi'(S)+(1+\e^\theta)\kappa.
\]
The same inequality holds with $\pi$ and $\pi'$ interchanged.
\end{restatable}
The proof is given in \Cref{app:privacy-proofs}.  The idea is to compare likelihood ratios on $\Omega_{\le D}$, where the two energies differ by at most $a\Delta$, and to bound the probability of $\Omega\setminus\Omega_{\le D}$ under the two Gibbs distributions by the resulting Gibbs tail bound.
\begin{restatable}[Approximate Differential Privacy Bound]{theorem}{ApproximatePrivacyUniformListTheorem}\label{thm:gibbs-privacy}\label{thm:main-privacy}
Assume the hypothesis of \Cref{thm:gibbs-over-lists}.  For $u\ge1$, let $Q^{(u)}$ be the Gibbs rule over all $K$-lists with inverse temperature $\eta_u\coloneqq(8/\gamma)(K\log m+u+1)$.  This rule satisfies
\[
   \Dist(Q^{(u)})
   \le
   3-\frac{7\gamma}{4}.
\]
Moreover, if $n\ge ({8}/{\gamma})A_K(1-{\gamma}/{2})$ for $A_K$ given in \Cref{prop:local-profile-ratio-change}, then the released-winner mechanism associated with $Q^{(u)}$ is $(\varepsilon_{\mathsf{DP}},\e^{-u})$-differentially private under replacement adjacency with
\[
   \varepsilon_{\mathsf{DP}}
   \le
   \frac{16}{\gamma}A_K\!\left(1-\frac{\gamma}{2}\right)
   \frac{K\log m+u+1}{n}.
\]
In particular, by \Cref{prop:br-uniform-list}, for every $\delta\in(0,1)$ there is a randomized ordinal voting rule $Q_\delta$, together with absolute constants $\varepsilon,C>0$, such that, whenever $n\ge C$, the released-winner mechanism associated with $Q_\delta$ is $(\varepsilon_{\mathsf{DP}},\delta)$-differentially private and satisfies
\[
   \Dist(Q_\delta)\le 3-\varepsilon,\qquad
   \varepsilon_{\mathsf{DP}}\le C\frac{\log m+\log(1/\delta)+1}{n}.
\]
\end{restatable}
The proof, given in \Cref{app:privacy-proofs}, reuses the low-energy sublevel and Gibbs-tail decomposition as in \Cref{sec:sensitivity}, but applies it to event probabilities.  The distortion part follows from the same Gibbs variational bound.  For privacy, with $E(L)=\lambda_\sigma(L)$ and $E'(L)=\lambda_\tau(L)$ for replacement-neighboring profiles $\sigma,\tau$, \Cref{prop:local-profile-ratio-change} gives the local drift bound on the sublevel $\min\{E(L),E'(L)\}\le D$, where $D=1-\gamma/2$, allowing \Cref{lem:gibbs-privacy-comparison} to apply; the complement of this sublevel is controlled by the Gibbs tail term in that lemma.

The same rule $Q_\delta$ inherits from \Cref{thm:gibbs-over-lists} the worst-case sensitivity bound
\[
   \Sens_{m,n}(Q_\delta)
   \le
   C'(\log m+\log(1/\delta)+1)/n
\]
for an absolute constant $C'$, as stated in \cref{tab:summary-results}.
\section*{Acknowledgments}
Shinsaku Sakaue was supported by JST BOOST Program Japan Grant Number \mbox{JPMJBY24D1}.
Kaito Fujii was supported by JSPS KAKENHI Grant Number \mbox{JP22K17857}.
Yuichi Yoshida was supported by JSPS KAKENHI Grant Number \mbox{JP24K02903}.
\printbibliography[heading=bibintoc]
\appendix
\crefalias{section}{appendix}
\mainbodylabelsfalse
\section{Deferred Details}\label[appendix]{app:deferred-details}
This section presents conventions, notation details, and source translations used by the main text.
\subsection{Normalization Conventions for Degenerate Instances}\label[appendix]{app:normalization-conventions}
The distortion definition keeps zero-optimum instances by using the extended-ratio convention stated in \Cref{sec:prelim}: a ratio with denominator zero is read as $0$ when the numerator is zero and as $+\infty$ otherwise.  Thus, if a rule has finite distortion and $\OPT_{d,\sigma}=0$, then $\SC_{d,\sigma}(F(\sigma))=0$.
For sensitivity, the definition in \cref{subsec:sensitivity} restricts the normalization to instances with $\OPT_{d,\sigma}>0$.  This restriction only removes degenerate normalized terms for finite-distortion rules.  Indeed, for any fixed pseudometric $d\consistent\sigma$, if $\OPT_{d,\sigma}=0$, then every candidate in the support of $F(\sigma)$ has zero social cost.  After deleting a voter, the same restricted metric still has zero optimum.  Let $c$ be a zero-cost candidate for the original profile and let $c'$ be a zero-cost candidate for the profile after deletion.  For every remaining voter $u$, both $d(u,c)$ and $d(u,c')$ are zero, so the triangle inequality gives $d(c,c')\le d(c,u)+d(u,c')=0$.  Hence $\W_d(F(\sigma),F(\sigma^{-v}))=0$ for every deletion $v$.  Thus every deletion term from an instance with $\OPT_{d,\sigma}=0$ is zero under the natural extension that assigns $0/0$ the value $0$, so including these degenerate instances does not affect the supremum defining sensitivity. The same argument applies to average sensitivity defined in \cref{app:baselines}.
\subsection{Details for the Biased-Metric Characterization}\label[appendix]{app:biased-ratio-details}
We relate \Cref{prop:biased-ratio} to the biased-metric characterization of
\citet[Theorem~2.8]{Cai2026}.  Their statement fixes a profile, a lottery $D$ over
candidates, and a nonnegative vector $x$ with $x_{i^*}=0$.  Their threshold set
$I_t=\{a:x_a\le t\}$ is our $I_t(x)$, their term $\ell(D,t)$ is
$\ell_{\sigma,p,x}(t)$, and their term $r(t)$ is $r_{\sigma,x}(t)$.
Their theorem states that the distortion bound $1+2\lambda$ is equivalent to
the inequalities
\[
   L_\sigma(p;x)\le \lambda R_\sigma(x)
\]
for all such vectors $x$.  Therefore, for a fixed profile and lottery, the
smallest admissible excess parameter is
\[
   \sup_x \frac{L_\sigma(p;x)}{R_\sigma(x)},
\]
with the extended-ratio convention used in \Cref{prop:biased-ratio}.
If $R_\sigma(x)=0<L_\sigma(p;x)$, then the biased metric induced by $x$ has
zero optimal social cost and positive lottery cost, so the distortion is
infinite.  If both terms are zero, the vector contributes zero to the excess
ratio.  Restricting to $\X$ only fixes scale: nonzero vectors can be rescaled
without changing the ratio, and the all-zero vector does not affect the
supremum.  \Citet{Cai2026} formulate the model using pseudometrics on the labeled
voters and candidates, which is equivalent to the possibly non-injective
embedding convention in \Cref{sec:prelim}.
\subsection{Bounded-Randomness Consequence and Constants}\label[appendix]{app:br-consequence}
This subsection gives the details behind \Cref{prop:br-uniform-list} and derives a conservative numerical lower bound on the gap.  All theorem, corollary, remark, and claim numbers in this subsection refer to arXiv:2602.08871v1 of \citet{Cai2026}.  In the cited arXiv version of \citet{Cai2026}, Theorem~7.6 gives, for each fixed profile, a real-valued mixture of constant-support rules with distortion below $3$ by an absolute margin.  Corollary~7.7 bounds the total support size by an absolute constant, uniformly over the profile and over $m,n$.  Remark~7.8 observes that a rational mixing probability realizes the resulting rule as selecting uniformly at random from a constant-size list.  The list form used in \Cref{prop:br-uniform-list} follows by replacing the fixed mixing probability by a rational number with denominator independent of the profile and of $m,n$, and then concatenating equal-length component lists.  Thus the conclusion is profile-wise: for every profile $\sigma$, there exists a constant-size list $L$ whose uniform lottery has distortion below $3$.

We now instantiate the constants from \citet[Appendix~B]{Cai2026} to obtain a conservative explicit lower bound on the final improvement over distortion $3$.  In the proof of \citet[Theorem~7.6]{Cai2026}, the parameters are chosen as
\[
   \alpha=\frac{L_k}{24},
   \widetilde\beta=\frac{L_k}{9},
   \quad
   \varepsilon_1=\frac{L_k^3}{150000},
   \varepsilon_2=\frac1k,
   \mu=1-\frac{L_k^2}{2000},
\]
where $L_k=\log(k/4)/k$ and $k\ge7$ is an integer.
For the gap calculation below, set
\[
   a_k=\frac{28L_k^3}{150000},
   b_k=\frac{L_k^2}{108},
   c_k=\frac{L_k}{2}-L_k^2,
   \mu_k=1-\frac{L_k^2}{2000}.
\]
The four derived quantities correspond to specific steps in \citet[Appendix~B]{Cai2026}.  The term $a_k=28\varepsilon_1$ is the approximation penalty for the Maximal Lottery term in Claims~B.4 and~B.5.  The term $b_k=2\alpha\widetilde\beta$ is the improvement in the not strongly consistent case in Claim~B.5, while $c_k=L_k/2-L_k^2$ is the strongly consistent gap supplied by Claim~B.3 and used in Claim~B.4.  Finally, $\mu_k=\mu$ is the mixture parameter from Theorem~7.6 before rational rounding.  With these identifications, the proofs of Claims~B.4 and~B.5 show that the mixed rule's distortion is below $3$ by at least the following margins in the strongly consistent and not strongly consistent cases:
\[
   \zeta_{\mathsf{str}}(k)
   =(1-\mu_k)c_k-\mu_k a_k,
   \zeta_{\mathsf{inc}}(k)
   =\mu_k(b_k-a_k)-12(1-\mu_k).
\]
Thus, before rationally rounding the mixing probability, the bounded-randomness gap is at least
\[
   \zeta_{\mathsf{fixed}}(k)
   =
   \min\{\zeta_{\mathsf{str}}(k),\zeta_{\mathsf{inc}}(k)\}.
\]
We take $k=11$, the integer nearest to the maximizer $4e$ of $L_k=\log(k/4)/k$.
The two values are
\[
   \zeta_{\mathsf{str}}(11)
   \approx
   1.3496019707\times10^{-8},
   \zeta_{\mathsf{inc}}(11)
   \approx
   2.7419103233\times10^{-5}.
\]
The strongly consistent case is the bottleneck, so $\zeta_{\mathsf{fixed}}(11)\approx1.3496019707\times10^{-8}$ before rational rounding.

To obtain the list form in \Cref{prop:br-uniform-list}, choose
\[
   q=\left\lceil\frac{16}{\zeta_{\mathsf{fixed}}(11)}\right\rceil,
   \qquad
   r=\left\lfloor q\mu_{11}+\frac12\right\rfloor .
\]
Then we have
\[
   \left|\frac rq-\mu_{11}\right|
   \le \frac{1}{2q}
   \le
   \frac{\zeta_{\mathsf{fixed}}(11)}{32}.
\]
Let $K_{\mathsf{ML}}$ and $K_{\mathsf{PL}}$ be fixed list lengths realizing the RepApx Maximal Lottery and RepApx Pruned Lottery components for the fixed parameters above; the latter uses the small-support Stable Lottery guarantee of \citet[Theorem~3.4]{Cai2026} after the deterministic pruning step.  These lengths are absolute constants.
Take $K_0$ to be any common multiple, for instance $K_0=K_{\mathsf{ML}}K_{\mathsf{PL}}$.  Repeating a list as a whole does not change its induced uniform lottery, so both component lists may be viewed as lists of length $K_0$.  Concatenating $r$ copies of the Maximal Lottery list and $q-r$ copies of the Pruned Lottery list gives one list of length
\[
   K_{\mathsf{BR}}\coloneqq qK_0,
\]
and a uniform draw from this concatenated list is exactly the mixture with weights $r/q$ and $1-r/q$.  Since the value above gives $q\ge10^9$, this conservative calculation makes $K_{\mathsf{BR}}$ at least $10^9K_0$.  As functions of the mixing probability $p$, the two margins have slopes $-(a_{11}+c_{11})$ and $b_{11}-a_{11}+12$, whose absolute values are at most $16$ for this parameter choice.  Hence replacing $\mu_{11}$ by $r/q$ reduces either margin by at most $16\zeta_{\mathsf{fixed}}(11)/32=\zeta_{\mathsf{fixed}}(11)/2$.  We fix the constant in \Cref{prop:br-uniform-list} by setting
\[
   \varepsilon_{\mathsf{BR}}\coloneqq\frac{\zeta_{\mathsf{fixed}}(11)}{2}.
\]
The improvement $\varepsilon_{\mathsf{BR}}$ in the distortion bound corresponds to the ratio slack $\gamma=\varepsilon_{\mathsf{BR}}/2$, and \Cref{thm:gibbs-over-lists} then shows that the final improvement in the distortion bound is
\[
   \frac{7\gamma}{4}
   =
   \frac{7\varepsilon_{\mathsf{BR}}}{8}
   \ge
   \frac{7}{16}\zeta_{\mathsf{fixed}}(11)
   >
   5.9\times10^{-9}.
\]
The main theorems only require the existence of some absolute positive gap; this numerical value is a conservative instantiation from arXiv v1.
\section{Deferred Proofs}\label[appendix]{app:deferred-proofs}
This section gives the proofs deferred from \Cref{sec:biased-ratio,sec:list-stability,sec:gibbs,sec:sensitivity,sec:approx-privacy}.
\subsection{Proof of the Support-Normalized Reduction}\label[appendix]{app:support-normalized-proof}
\SupportNormalizedReduction*
\begin{proof}
The inequality ``$\ge$'' is immediate.  For the other direction, fix $x\in\X$ and put
\[
   h=\max_{j\in\Supp(p_L)}x_j.
\]
If $h=0$, then every candidate with positive $p_L$-mass has $x$-value $0$,
and hence belongs to $I_t(x)$ for every $t\ge0$.  Therefore no term with
$j\notin I_t(x)$ and $p_L(j)>0$ appears in the sum defining
$\ell_{\sigma,p_L,x}(t)$.
Hence $\ell_{\sigma,p_L,x}(t)\equiv0$ and $\lambda_\sigma(L;x)=0$.  Since $|\Cands|\ge2$, the set $\Xsn(L)$ is nonempty, and all fixed-vector ratios are nonnegative; hence this zero value does not affect the supremum.

Assume $h>0$ and define
\[
   y_c\coloneqq\frac{\min\{x_c,h\}}{h},
   \qquad c\in\Cands.
\]
Then $y\in\Xsn(L)$ holds: all coordinates of $y$ lie in $[0,1]$; since $x\in\X$,
some coordinate has $x$-value $0$ and hence $y$-value $0$; by the definition
of $h$, some support candidate has $y$-value $1$.  Thus $y\in\X$ and
$\max_{c\in\Supp(p_L)}y_c=1$.

For every $t\in[0,1)$, we claim
\[
   I_t(y)=I_{ht}(x).
\]
Indeed, if $x_a\le h$, then $y_a\le t$ is equivalent to $x_a\le ht$; if $x_a>h$, then $y_a=1>t$ and also $x_a>h>ht$.

The level-set identity allows the numerator for $y$ to be compared with the numerator for $x$ after the change of variables $u=ht$.  Specifically, for every
$0\le t<1$, we have
\[
   \ell_{\sigma,p_L,y}(t)
   =
   \sum_{j\notin I_t(y)}
   s_\sigma(I_t(y)\succ j)p_L(j)
   =
   \sum_{j\notin I_{ht}(x)}
   s_\sigma(I_{ht}(x)\succ j)p_L(j)
   =
   \ell_{\sigma,p_L,x}(ht).
\]
Moreover, for $u\ge h$ every support candidate lies in $I_u(x)$, so $\ell_{\sigma,p_L,x}(u)=0$.  Hence the change of variables $u=ht$ yields
\[
   L_\sigma(p_L;y)
   =\int_0^1\ell_{\sigma,p_L,x}(ht)\,\diff t
   =\frac1h\int_0^h\ell_{\sigma,p_L,x}(u)\,\diff u
   =\frac1h L_\sigma(p_L;x).
\]
For the denominator, the clipping map $z\mapsto\min\{z,h\}$ is nondecreasing and $1$-Lipschitz.  Thus, for every ordered pair $a\succ_v b$, we have
\[
   (y_a-y_b)_+
   =\frac{(\min\{x_a,h\}-\min\{x_b,h\})_+}{h}
   \le \frac{(x_a-x_b)_+}{h}.
\]
Taking the maximum over pairs and then averaging over voters gives the denominator comparison
\[
   R_\sigma(y)\le \frac1h R_\sigma(x).
\]
If $R_\sigma(x)=0$, then the denominator inequality gives $R_\sigma(y)=0$ as well, since $R_\sigma(y)\ge0$.  By the extended-ratio convention, $\lambda_\sigma(L;x)$ is either $0$ or $+\infty$ according as $L_\sigma(p_L;x)$ is zero or positive.  The numerator identity $L_\sigma(p_L;y)=L_\sigma(p_L;x)/h$ preserves this dichotomy, so $\lambda_\sigma(L;y)=\lambda_\sigma(L;x)$ in both subcases.  If $R_\sigma(x)>0$, we have
\[
   \lambda_\sigma(L;y)
   =\frac{L_\sigma(p_L;y)}{R_\sigma(y)}
   \ge
   \frac{L_\sigma(p_L;x)/h}{R_\sigma(x)/h}
   =\lambda_\sigma(L;x).
\]
Thus every vector can be replaced by a support-normalized vector of at least the same value.
\end{proof}
\subsection{Proof of the Ratio Change Bound for Local Profile Changes}\label[appendix]{app:local-profile-ratio-change-proof}
\LocalProfileRatioChangeBound*
\begin{proof}
Assume first that $\Lambda\coloneqq\lambda_\sigma(L)<\infty$ and $\Delta\le 1/(4(K\Lambda+1))$.  By \Cref{lem:support-normalized-reduction}, both values $\lambda_\sigma(L)$ and $\lambda_{\sigma'}(L)$ are suprema over the same set $\Xsn(L)$, which depends only on the list.  Fix $x\in\Xsn(L)$.  Since $\lambda_\sigma(L;x)\le \lambda_\sigma(L)=\Lambda$, \Cref{lem:denominator-lower-bound} gives
\[
   R\coloneqq R_\sigma(x)\ge \underline R_{K,\Lambda}=\frac{1}{2(K\Lambda+1)}.
\]
The assumption on $\Delta$ is exactly $\Delta\le \underline R_{K,\Lambda}/2$.  Let $L_0\coloneqq L_\sigma(p_L;x)$.  Since $R>0$ and $\lambda_\sigma(L;x)\le \Lambda$, we have $L_0\le \Lambda R$.  By \eqref{eq:local-profile-average-drift}, the profile $\sigma'$ satisfies
\[
   R_{\sigma'}(x)
   \ge R-\Delta
   \ge \underline R_{K,\Lambda}/2,
   \qquad
   L_{\sigma'}(p_L;x)
   \le L_0+\Delta.
\]
Thus the ratio associated with $x$ on $\sigma'$ is finite, and
\[
   \lambda_{\sigma'}(L;x)-\lambda_\sigma(L;x)
   \le
   \frac{L_0+\Delta}{R-\Delta}-\frac{L_0}{R}
   =
   \frac{\Delta(R+L_0)}{R(R-\Delta)}
   \le
   \frac{2(\Lambda+1)}{\underline R_{K,\Lambda}}\Delta
   =A_K(\Lambda)\Delta.
\]
For the reverse direction on the same fixed vector, \eqref{eq:local-profile-average-drift} also gives $L_{\sigma'}(p_L;x)\ge L_0-\Delta$ and $R_{\sigma'}(x)\le R+\Delta$.  Hence
\[
   \lambda_\sigma(L;x)-\lambda_{\sigma'}(L;x)
   \le
   \frac{L_0}{R}-\frac{L_0-\Delta}{R+\Delta}
   =
   \frac{\Delta(R+L_0)}{R(R+\Delta)}
   \le
   \frac{\Lambda+1}{\underline R_{K,\Lambda}}\Delta
   \le A_K(\Lambda)\Delta.
\]
Therefore $|\lambda_{\sigma'}(L;x)-\lambda_\sigma(L;x)|\le A_K(\Lambda)\Delta$ for every $x\in\Xsn(L)$.  This bound also gives $\lambda_{\sigma'}(L;x)\le\Lambda+A_K(\Lambda)\Delta$ uniformly in $x\in\Xsn(L)$, so $\lambda_{\sigma'}(L)<\infty$.
Taking suprema in the above two inequalities gives
\[
   |\lambda_{\sigma'}(L)-\lambda_\sigma(L)|
   \le
   A_K(\Lambda)\Delta.
\]
For the reverse implication, apply the same calculation with the roles of $\sigma$ and $\sigma'$ interchanged.  Since the additive perturbation condition \eqref{eq:local-profile-average-drift} is symmetric in $\sigma$ and $\sigma'$, the quotient calculation applies with $\lambda_{\sigma'}(L)$ in place of $\Lambda$ and gives the stated reverse bound.

For the one-voter instantiations, suppose $\sigma$ has $n$ voters and $\sigma'$ is obtained from $\sigma$ either by deleting one voter, with $n\ge2$, or by replacing one voter's ranking.  For each fixed $x\in\X$ and threshold $t$, the integrands defining $L_\sigma(p_L;x)$ and $R_\sigma(x)$ are averages of $[0,1]$-valued terms over voters.  The corresponding integrands for $\sigma'$ are obtained by deleting one term from these averages or by replacing one term.  The pointwise change is therefore at most $1/n$.  Since the integrands vanish for $t\ge1$, \eqref{eq:local-profile-average-drift} holds with $\Delta=1/n$, and the one-voter versions of the two implications follow.
\end{proof}
\subsection{Basic Properties of the Gibbs Rule}\label[appendix]{app:gibbs-axioms}
This subsection justifies \Cref{rem:gibbs-basic-properties}. Anonymity follows because $\lambda_\sigma(L)$ depends on the profile only through averages over voters and the ordinal constraints defining consistency; permuting the voters leaves these quantities unchanged.

For neutrality, a relabeling $\varphi$ of candidates induces a bijection
$L\mapsto\varphi(L)$ on $\Lists_K$.  The consistent pseudometrics and social
costs are relabeled in the same way, hence
$\lambda_{\varphi(\sigma)}(\varphi(L))=\lambda_\sigma(L)$.  Thus the Gibbs
distribution over lists is pushed forward by this bijection, and the induced
candidate lotteries satisfy $Q_{\varphi(\sigma)}^{(u)}(\varphi(c))=Q_\sigma^{(u)}(c)$ for every $c\in\Cands$.

For unanimity, suppose all voters rank a candidate $a$ first.  If a list $L$ assigns positive probability to a candidate $b\neq a$, consider the consistent pseudometric that co-locates all voters with $a$ and places $b$ at positive distance from them.  Then $\OPT_{d,\sigma}=0$ while $\SC_{d,\sigma}(p_L)>0$, so $L$ has infinite distortion and hence infinite energy.  The Gibbs distribution assigns zero mass to such lists, and the released candidate is $a$ with probability one.
\subsection{Proofs of the Finite-State Gibbs Estimates}\label[appendix]{app:gibbs-estimates-proofs}
We collect basic finite-state facts about Gibbs distributions, with
proofs included for completeness.
\begin{lemma}[Gibbs Variational Inequality with Infinite Energies]\label{lem:gibbs-variational}
Let $E\colon\Omega\to[0,\infty]$ be an energy on a finite set, and let $\pi_\eta$ be its Gibbs distribution at inverse temperature $\eta>0$.  If $\min_{\omega\in\Omega}E(\omega)\le B<\infty$, then
\[
   \E_{\omega\sim\pi_\eta}[E(\omega)]
   \le
   B+\frac{\log|\Omega|}{\eta}.
\]
\end{lemma}
\begin{proof}
Let $S=\Set{\omega\in\Omega}{E(\omega)<\infty}$ and put $Z=\sum_{\omega\in S}\e^{-\eta E(\omega)}$.  Since $\e^{-\eta E(\omega)}=0$ when $E(\omega)=\infty$, the Gibbs distribution assigns zero mass to $\Omega\setminus S$, and we interpret $\E_{\omega\sim\pi_\eta}[E(\omega)]$ as the expectation over $S$.  For every $\omega\in S$, we have
\[
   \log \pi_\eta(\omega)=-\eta E(\omega)-\log Z.
\]
Writing $H(\pi_\eta)=-\sum_{\omega\in S}\pi_\eta(\omega)\log\pi_\eta(\omega)$ for the Shannon entropy of $\pi_\eta$, we have
\[
   \E_{\omega\sim\pi_\eta}[E(\omega)]
   =\frac{H(\pi_\eta)-\log Z}{\eta}
   \le
   \frac{\log|\Omega|-\log Z}{\eta}.
\]
Choose $\omega_0$ with $E(\omega_0)\le B$.  Then $Z\ge \e^{-\eta B}$ holds, and the claimed bound follows.
\end{proof}
\begin{lemma}[Gibbs Moment and Tail Bounds]\label{lem:gibbs-moments}
Fix $B\ge0$ and an integer $r\ge0$.  Let $E\colon\Omega\to[0,\infty]$ be an energy on a finite set with $\min_{\omega\in\Omega}E(\omega)\le B$, and let $\pi_\eta$ be its Gibbs distribution at inverse temperature $\eta\ge1$, i.e., $\pi_\eta(\omega)=\e^{-\eta E(\omega)}/Z$ where $Z=\sum_{\omega\in\Omega}\e^{-\eta E(\omega)}$.
If $\eta\ge\log|\Omega|$, then there is a constant $C_{r,B}<\infty$ depending only on $r$ and $B$ such that
\[
   \E_{\omega\sim\pi_\eta}[(1+E(\omega))^r]
   \le
   C_{r,B}.
\]
Moreover, for every $M\ge0$ such that $(1+t)^r \e^{-\eta t}$ is nonincreasing on $[M,\infty)$, the weighted tail satisfies
\[
   \sum_{\substack{\omega:M<E(\omega)<\infty}}
   \pi_\eta(\omega)(1+E(\omega))^r
   \le
   |\Omega|(1+M)^r\exp(-\eta(M-B)).
\]
\end{lemma}
\begin{proof}
The latter tail estimate follows from $Z\ge \e^{-\eta B}$ and the monotonicity of $(1+t)^r \e^{-\eta t}$ on $[M,\infty)$.  For the moment bound, put $M_{r,B}\coloneqq B+r+2$.  This choice satisfies the monotonicity condition in that estimate: for $t\ge M_{r,B}$, the derivative of $\log((1+t)^r \e^{-\eta t})$ is $r/(1+t)-\eta\le0$.
Split the expectation according to whether $E(\omega)\le M_{r,B}$.  The contribution from $E(\omega)\le M_{r,B}$ is at most $(1+M_{r,B})^r$.  For the contribution from $E(\omega)>M_{r,B}$, the tail estimate with $M=M_{r,B}$ gives
\[
   \sum_{\substack{\omega:M_{r,B}<E(\omega)<\infty}}
   \pi_\eta(\omega)(1+E(\omega))^r
   \le
   |\Omega|(1+M_{r,B})^r\exp(-\eta(M_{r,B}-B)).
\]
Using $\log|\Omega|\le\eta$, $M_{r,B}-B=r+2$, and $\eta\ge1$, the right-hand side is at most
\[
   (1+M_{r,B})^r\exp(-\eta(r+1))
   \le
   (1+M_{r,B})^r \e^{-(r+1)}.
\]
Thus, the full moment is bounded by a constant depending only on $r$ and $B$.
\end{proof}
\subsection{Proofs for Sensitivity}\label[appendix]{app:gibbs-over-lists-proof}
\EnergyDependentWeightedGibbsStability*
\begin{proof}
Let $Z,Z'$ be the partition functions for $E,E'$.  Put $C_K\coloneqq4K$; since $K\ge1$, we have $A_K(t)=4(Kt+1)(t+1)\le C_K(1+t)^2$ for all $t\ge0$, and hence $(1+t)(1+A_K(t))$ is bounded by a constant depending only on $K$ times $(1+t)^3$.  By the moment part of \Cref{lem:gibbs-moments} with $r=2,3$, there is a
constant $D_{K,B}\ge1$ with the following property.  Whenever $G$ is an
energy with $\min_{\omega\in\Omega}G(\omega)\le B$ and $\mu$ is its Gibbs law
at inverse temperature $\eta\ge\max\{1,\log|\Omega|\}$, we have
\begin{equation}\label{eq:moment-control}
   \E_{\omega\sim\mu}\bigl[A_K(G(\omega))\bigr]
   \le D_{K,B},
   \E_{\omega\sim\mu}
   \bigl[(1+G(\omega))(1+A_K(G(\omega)))\bigr]
   \le D_{K,B}.
\end{equation}
Write $s\coloneqq\eta\Delta$ and set
\[
   M=(16C_Ks)^{-1/2}.
\]
Choose $\nu_{K,B}>0$ so that, whenever $0<s\le \nu_{K,B}$, the cutoff $M=(16C_Ks)^{-1/2}$ satisfies
\begin{equation}\label{eq:smallness-conditions}
   M\ge B+3,
   \qquad
   s\le \frac{1}{4(KM+1)},
   \qquad
   (1+M)\e^{-(M-B-1)}\le s,
   \qquad
   D_{K,B}s\le \frac1{12}.
\end{equation}
Such a choice exists because, as $s\downarrow0$, the cutoff
$M=(16C_Ks)^{-1/2}$ tends to infinity, while
$s(KM+1)\to0$, $D_{K,B}s\to0$, and the exponential factor
$\e^{-M}$ dominates the polynomial factor $(1+M)$ and the factor $s^{-1}$. For completeness, we give one quantitative sufficient choice.
\paragraph{Concrete choice of $\nu_{K,B}$.}
A quantitative sufficient choice is obtained from these conditions as follows.  Put $H_{K,B}\coloneqq B+\log K+1$.  Choose a universal constant $c_0\in(0,1)$ and write $\beta=(8\sqrt{c_0})^{-1}$, with $c_0$ small enough that $c_0\le1/12$, $\sqrt{c_0}/8+c_0\le1/4$, $\beta\ge3$, and
\[
   (\beta-1)x\ge \log 64+3\log(1+\beta x)
   \qquad\text{for every }x\ge1.
\]
Such a universal choice exists: for example, take $\beta=100$ and $c_0=1/(64\beta^2)$; then $\beta=(8\sqrt{c_0})^{-1}$, then the above conditions on $c_0$ and $\beta$ hold, and the function $(\beta-1)x-3\log(1+\beta x)$ is increasing on $x\ge1$ with value at $x=1$ larger than $\log 64$.  We may take
\[
   \nu_{K,B}
   =
   c_0\min\left\{
      D_{K,B}^{-1},
      \frac{1}{K H_{K,B}^2}
   \right\}.
\]
Indeed, if $0<s\le\nu_{K,B}$, we have $D_{K,B}s\le1/12$, $M=(64Ks)^{-1/2}\ge\beta H_{K,B}$, and $s\le c_0$ and $Ks\le c_0/H_{K,B}^2\le c_0$.  Thus
\[
   s(KM+1)=\frac{\sqrt{Ks}}8+s
   \le \frac{\sqrt{c_0}}8+c_0
   \le \frac14.
\]
The bound $M\ge\beta H_{K,B}$ and $\beta\ge3$ also give $M\ge B+3$.  For the exponential tail condition, the identity $s=1/(64KM^2)$ shows that it is enough to prove
\[
   M\ge B+1+\log(64KM^2(1+M)).
\]
Since $M\ge\beta H_{K,B}\ge3$, the function $x\mapsto x-3\log(1+x)$ is increasing on $[\beta H_{K,B},\infty)$.  Hence
\[
   M-3\log(1+M)
   \ge
   \beta H_{K,B}-3\log(1+\beta H_{K,B})
   \ge
   H_{K,B}+\log 64.
\]
Using $H_{K,B}=B+\log K+1$ and $2\log M+\log(1+M)\le3\log(1+M)$ gives the desired exponential tail condition.  Thus the presented scale is a valid sufficient choice for $\nu_{K,B}$.
\paragraph{Sublevel and tail controls.}
The first two inequalities in \eqref{eq:smallness-conditions} imply that, for every $t\le M$, the following sublevel controls hold:
\begin{equation}\label{eq:sublevel-weight-control}
   \eta A_K(t)\Delta=sA_K(t)\le sC_K(1+M)^2\le4sC_KM^2=\frac14.
\end{equation}
\begin{equation}\label{eq:sublevel-drift-condition}
   \Delta=\frac{s}{\eta}\le s\le\frac{1}{4(KM+1)}\le\frac{1}{4(Kt+1)}.
\end{equation}
Since $M\ge B+3\ge1$ and $\eta\ge1$, the function $(1+t)^r \e^{-\eta t}$ is nonincreasing on $[M,\infty)$ for $r=0,1$.  The weighted tail estimate in \Cref{lem:gibbs-moments}, together with $\log|\Omega|\le\eta$, gives the following bound for either energy $G\in\set{E,E'}$ and its Gibbs law $\mu$, for $r=0,1$:
\begin{equation}\label{eq:tail-control}
   \sum_{\substack{\omega:M<G(\omega)<\infty}}
   \mu(\omega)(1+G(\omega))^r
   \le
   |\Omega|(1+M)^r \e^{-\eta(M-B)}
   \le
   (1+M)\e^{-(M-B-1)}
   \le s.
\end{equation}
\paragraph{Partition-function comparison.}
We now use the sublevel and tail bounds to compare the partition functions.  For states $\omega$ with $E(\omega)\le M$, \eqref{eq:sublevel-drift-condition} with $t=E(\omega)$ lets us apply the local drift condition, while \eqref{eq:sublevel-weight-control} with $t=E(\omega)$ gives $\eta A_K(E(\omega))\Delta\le1/4$.  Hence $E'(\omega)<\infty$, and the exponential weights satisfy
\[
   \e^{-\eta E'(\omega)}
   \ge
   \e^{-\eta E(\omega)}\bigl(1-2\eta A_K(E(\omega))\Delta\bigr),
\]
where we use $\e^{-z}\ge1-2z$ for $0\le z\le1/4$.  Keeping only the terms with $E(\omega)\le M$ in $Z'$ and normalizing by $Z$ gives
\begin{equation}\label{eq:partition-lower}
   \frac{Z'}{Z}\ge \sum_{\omega:E(\omega)\le M}\pi(\omega)\bigl(1-2\eta\Delta A_K(E(\omega))\bigr)\ge 1-\pi(\Set*{\omega}{E(\omega)>M})-2\eta\Delta\,\E_{\omega\sim\pi}\bigl[A_K(E(\omega))\bigr].
\end{equation}
In \eqref{eq:partition-lower}, the tail term is at most $s$ by \eqref{eq:tail-control} with $G=E,\mu=\pi,r=0$, and the moment term is at most $2D_{K,B}s$ by \eqref{eq:moment-control} with $G=E,\mu=\pi$.  Since $D_{K,B}\ge1$, this gives $Z'/Z\ge1-3D_{K,B}s$.  The same argument with $E$ and $E'$ interchanged gives $Z/Z'\ge1-3D_{K,B}s$.  Since $3D_{K,B}s\le1/4$ by \eqref{eq:smallness-conditions}, applying $-\log(1-a)\le2a$ with $a=3D_{K,B}s \in (0,1/4]$ to the two inequalities gives the partition-ratio bound
\begin{equation}\label{eq:partition-ratio}
   \left|\log\frac{Z'}{Z}\right|
   \le
   6D_{K,B}s.
\end{equation}
\paragraph{Derivation of \eqref{eq:energy-weighted-positive-part-bound}.}
We now bound the first sum on the left-hand side of \eqref{eq:energy-weighted-positive-part-bound}.  The contribution from states with $M<E(\omega)<\infty$ is at most
\[
   \sum_{\omega:M<E(\omega)<\infty}\pi(\omega)(1+E(\omega))
   \le
   s
\]
by \eqref{eq:tail-control} with $G=E,\mu=\pi,r=1$.  For states $\omega$ with $E(\omega)\le M$, \eqref{eq:sublevel-drift-condition}
with $t=E(\omega)$ gives the smallness condition required to apply the local drift assumption; together with \eqref{eq:partition-ratio}, this gives
\[
   \left|\log\frac{\pi(\omega)}{\pi'(\omega)}\right|
   =
   \left|-\eta(E(\omega)-E'(\omega))+\log\frac{Z'}{Z}\right|
   \le
   s A_K(E(\omega))+6D_{K,B}s.
\]
By \eqref{eq:sublevel-weight-control} with $t=E(\omega)$ and $D_{K,B}s\le1/12$ from \eqref{eq:smallness-conditions}, the last expression is
at most $1/4+1/2<1$.  Thus, with
$z=\log(\pi(\omega)/\pi'(\omega))$, we have $|z|<1$.  Since
$(\pi(\omega)-\pi'(\omega))_+=\pi(\omega)(1-\e^{-z})_+$ and
$(1-\e^{-z})_+\le2|z|$ for $|z|\le1$, it follows that, for $E(\omega)\le M$,
\begin{equation}\label{eq:positive-part-bound}
   (\pi(\omega)-\pi'(\omega))_+
   \le
   2\pi(\omega)s\bigl(A_K(E(\omega))+6D_{K,B}\bigr).
\end{equation}
Multiplying \eqref{eq:positive-part-bound} by $1+E(\omega)$ and summing over
$E(\omega)\le M$ gives
\[
   \sum_{\omega:E(\omega)\le M}
   (\pi(\omega)-\pi'(\omega))_+(1+E(\omega))
      \le{}
      2s\,\E_{\omega\sim\pi} \bigl[(1+E(\omega))A_K(E(\omega))\bigr]
   +12D_{K,B}s\,\E_{\omega\sim\pi}[1+E(\omega)].
\]
The two expectations are bounded by constants depending only on $K$ and $B$
by \eqref{eq:moment-control}.  Hence, choosing
$\widetilde C_{K,B}$ large enough, we have
\[
   \sum_{\omega:E(\omega)\le M}
   (\pi(\omega)-\pi'(\omega))_+(1+E(\omega))
   \le
   \widetilde C_{K,B}\eta\Delta.
\]
The second sum is handled in the same way after interchanging
$E,\pi,Z$ with $E',\pi',Z'$; the symmetric drift assumption gives the
pointwise comparison, and \eqref{eq:partition-ratio} controls the partition
functions in both directions.
Summing the two bounds proves the lemma with $C_{K,B}=2\widetilde C_{K,B}$.
\end{proof}
We also use the following elementary transport bound to convert list-distribution stability into a bound on the Wasserstein term in sensitivity.
\begin{lemma}[Transport through an Optimum]\label{lem:transport-through-opt}
For any lotteries $p,q\in\triangle(\Cands)$ and any candidate $o\in\Cands$, transporting all mass through $o$ gives
\[
   \W_d(p,q)
   \le
   \sum_c p(c)d(c,o)+\sum_c q(c)d(c,o).
\]
For every profile $\sigma$, lottery $p$, and candidate $o$, we also have
\[
   \sum_c p(c)d(c,o)
   \le
   \SC_{d,\sigma}(p)+\SC_{d,\sigma}(o).
\]
In particular, if $o\in\argmin_c\SC_{d,\sigma}(c)$, then the right-hand side is $\SC_{d,\sigma}(p)+\OPT_{d,\sigma}$.
\end{lemma}
\begin{proof}
The first bound comes from the feasible transport that moves all mass of $p$ to $o$ and then from $o$ to $q$.  For the second, the triangle inequality gives, for each candidate $c$, the bound
\[
   d(c,o)
   \le
   \frac1n\sum_{v\in\Voters}\bigl(d(c,v)+d(v,o)\bigr)
   =
   \SC_{d,\sigma}(c)+\SC_{d,\sigma}(o).
\]
Averaging over $p$ proves the general form, and the optimum-specialized form follows from $\SC_{d,\sigma}(o)=\OPT_{d,\sigma}$.
\end{proof}
\AverageSensitivityUniformListTheorem*
\begin{proof}
If $m=1$, the unique candidate has distortion $1$ and sensitivity $0$.  Assume $m\ge2$.  The sensitivity statement is only for $n\ge2$.
\paragraph{Distortion.}
For any profile, apply \Cref{lem:gibbs-variational} to $E(L)=\lambda_\sigma(L)$ and the common inverse temperature $\eta_u$.  Since $|\Lists_K|\le m^K$, we have
\[
   \E_{L\sim\Pi_\sigma^{(\eta_u)}}[\lambda_\sigma(L)]
   \le
   1-\gamma+\frac{K\log m}{\eta_u}
   \le
   1-\frac{7\gamma}{8}.
\]
The Gibbs distribution assigns zero mass to lists with infinite $\lambda_\sigma(L)$, and \Cref{prop:biased-ratio} gives $\rho_\sigma(L)=1+2\lambda_\sigma(L)$ on the remaining support.  The worst-case distortion value is convex in the lottery, hence
\[
   \rho_\sigma(Q^{(u)}(\sigma))
   \le
   \E_{L\sim\Pi_\sigma^{(\eta_u)}}[\rho_\sigma(L)]
   =
   1+2\E_{L\sim\Pi_\sigma^{(\eta_u)}}[\lambda_\sigma(L)]
   \le
   3-\frac{7\gamma}{4}.
\]
This holds for every profile, so the distortion bound follows.
\paragraph{Gibbs Stability.}
Fix an instance $(\sigma,d)$ with $n\ge2$ and a deletion $\sigma'=\sigma^{-v}$.  If $\OPT_{d,\sigma}=0$, then the zero-optimum discussion in \Cref{app:normalization-conventions}, applied to the finite-distortion rule $Q^{(u)}$ just proved above, gives $\W_d(Q^{(u)}(\sigma),Q^{(u)}(\sigma'))=0$. We may assume for the rest of the proof that $\OPT_{d,\sigma}>0$.

Define the energies by
\[
   E(L)=\lambda_\sigma(L),
   \qquad
   E'(L)=\lambda_{\sigma'}(L).
\]
Also let $\Pi=\Pi_\sigma^{(\eta_u)}$ and $\Pi'=\Pi_{\sigma'}^{(\eta_u)}$.  Put $B=1-\gamma$ and $\Delta=1/n$.  The existence assumption in the theorem holds for both $\sigma$ and $\sigma'$, so $\min_{L\in\Lists_K}E(L)\le B$ and $\min_{L\in\Lists_K}E'(L)\le B$.  The first two implications in \Cref{prop:local-profile-ratio-change} give exactly the local drift condition required in \Cref{lem:energy-dependent-gibbs}.  Also, we have
\[
   \frac{\log|\Lists_K|}{\eta_u}
   \le
   \frac{K\log m}{(8/\gamma)(K\log m+u+1)}
   \le1.
\]
For convenience, let $\mathcal T(\Pi,\Pi';E,E')$ denote the left-hand side of \eqref{eq:energy-weighted-positive-part-bound} with $(\Omega,\pi,\pi')=(\Lists_K,\Pi,\Pi')$.
If $\eta_u\Delta\le \nu_{K,B}$, \Cref{lem:energy-dependent-gibbs} gives, for a constant $C_{K,\gamma}^{\mathsf{Gibbs}}<\infty$ depending only on $K$ and $\gamma$, the bound
\begin{equation}\label{eq:gibbs-weighted-list-stability}
   \mathcal T(\Pi,\Pi';E,E')
   \le
   C_{K,\gamma}^{\mathsf{Gibbs}}\eta_u\Delta
   \le
   C_{K,\gamma}^{\mathsf{Gibbs}}\frac{K\log m+u+1}{n}.
\end{equation}
If $\eta_u\Delta>\nu_{K,B}$, then ${(K\log m+u+1)}/{n} > {\gamma \nu_{K,B}}/{8}$.
As shown below, this lower bound offers a way to control the Wasserstein distance in this remaining case.
\paragraph{From List Stability to Sensitivity.}
Let $o$ be an optimal candidate for the original instance $(\sigma,d)$, and let $\delta_o$ denote the point mass at $o$.  The same metric $d$ restricted to the remaining voters satisfies $d\consistent\sigma'$.  The candidate $o$ need not be optimal for $\sigma'$; it is used only in the transport comparison.  Nevertheless, we have
\begin{equation}\label{eq:deletion-opt-comparison}
   \OPT_{d,\sigma'}
   \le
   \SC_{d,\sigma'}(o)
   =\frac1{n-1}\sum_{w\ne v}d(w,o)
   \le
   \frac{n}{n-1}\OPT_{d,\sigma}.
\end{equation}
To bound the Wasserstein distance between the candidate lotteries, couple the common part of $\Pi$ and $\Pi'$ by the same list, move unmatched mass on the original side from $p_L$ to $\delta_o$, and move unmatched mass on the deletion side from $\delta_o$ to $p_L$.  For unmatched mass originating from $\Pi$, \Cref{lem:transport-through-opt} and the bound $\rho_\sigma(L)=1+2E(L)$ give
\[
   \frac{\W_d(p_L,\delta_o)}{\OPT_{d,\sigma}}
   \le
   \frac{\SC_{d,\sigma}(p_L)+\OPT_{d,\sigma}}{\OPT_{d,\sigma}}
   \le
   \rho_\sigma(L)+1
   =2+2E(L)
   \le4(1+E(L)).
\]
For unmatched mass originating from $\Pi'$, the same metric restricted to the remaining voters is feasible in the definition of $\rho_{\sigma'}(L)$.  Using the general form of \Cref{lem:transport-through-opt} and \eqref{eq:deletion-opt-comparison}, we obtain
\[
\begin{aligned}
   \frac{\W_d(p_L,\delta_o)}{\OPT_{d,\sigma}}
   &\le
   \frac{\SC_{d,\sigma'}(p_L)+\SC_{d,\sigma'}(o)}{\OPT_{d,\sigma}}
   =
   \frac{\SC_{d,\sigma'}(p_L)}{\OPT_{d,\sigma}}
   +\frac{\SC_{d,\sigma'}(o)}{\OPT_{d,\sigma}}
   \\
   &\le
   \rho_{\sigma'}(L)\frac{\OPT_{d,\sigma'}}{\OPT_{d,\sigma}}
   +\frac{\SC_{d,\sigma'}(o)}{\OPT_{d,\sigma}}
   \le
   \frac{n}{n-1}\bigl(\rho_{\sigma'}(L)+1\bigr)
   \le4(1+E'(L)).
\end{aligned}
\]
Since
\[
   Q^{(u)}(\sigma)=\sum_L\Pi(L)p_L,
   \qquad
   Q^{(u)}(\sigma')=\sum_L\Pi'(L)p_L,
\]
the list-level decomposition above induces a candidate-level transport plan:
the common mass $\min\{\Pi(L),\Pi'(L)\}$ is coupled through the same lottery
$p_L$, while the two unmatched parts are routed through $\delta_o$.
Thus the left and right marginals of this plan are
$Q^{(u)}(\sigma)$ and $Q^{(u)}(\sigma')$, respectively.  The common mass contributes no cost.  After normalizing by $\OPT_{d,\sigma}$, the unmatched mass originating from $\Pi$ contributes at most
\[
   4\sum_{L:E(L)<\infty}(\Pi(L)-\Pi'(L))_+(1+E(L)),
\]
and the unmatched mass originating from $\Pi'$ contributes at most
\[
   4\sum_{L:E'(L)<\infty}(\Pi'(L)-\Pi(L))_+(1+E'(L)).
\]
Set $C_{K,\gamma}^{\mathsf{small}}\coloneqq4C_{K,\gamma}^{\mathsf{Gibbs}}$.  By the definition of $\mathcal T$, these two bounds and \eqref{eq:gibbs-weighted-list-stability} give, whenever $\eta_u\Delta\le \nu_{K,B}$,
\[
   \frac{\W_d(Q^{(u)}(\sigma),Q^{(u)}(\sigma'))}{\OPT_{d,\sigma}}
   \le
   4\mathcal T(\Pi,\Pi';E,E')
   \le
   C_{K,\gamma}^{\mathsf{small}}\frac{K\log m+u+1}{n}.
\]
If $\eta_u\Delta>\nu_{K,B}$, use the distortion bound already proved.  Let $\rho_\gamma=3-7\gamma/4$.  \Cref{lem:transport-through-opt} and \eqref{eq:deletion-opt-comparison} give
\[
   \frac{\W_d(Q^{(u)}(\sigma),\delta_o)}{\OPT_{d,\sigma}}
   \le \rho_\gamma+1,
   \frac{\W_d(Q^{(u)}(\sigma'),\delta_o)}{\OPT_{d,\sigma}}
   \le \frac{n}{n-1}(\rho_\gamma+1)
   \le2(\rho_\gamma+1).
\]
The triangle inequality gives a normalized Wasserstein distance at most $3(\rho_\gamma+1)$.  Since $\eta_u=(8/\gamma)(K\log m+u+1)$ and $\Delta=1/n$, the inequality $\eta_u\Delta>\nu_{K,B}$ implies
\[
   \frac{K\log m+u+1}{n}
   >
   \frac{\gamma \nu_{K,B}}{8}.
\]
Define $C_{K,\gamma}^{\mathsf{large}} \coloneqq\frac{24(\rho_\gamma+1)}{\gamma \nu_{K,B}}$.
The preceding lower bound on $(K\log m+u+1)/n$ gives
\[
   3(\rho_\gamma+1)
   \le
   C_{K,\gamma}^{\mathsf{large}}\frac{K\log m+u+1}{n}.
\]
Taking $C_{K,\gamma}=\max\{C_{K,\gamma}^{\mathsf{small}},C_{K,\gamma}^{\mathsf{large}}\}$ proves the stated normalized Wasserstein bound for every deletion.
Taking the supremum over instances and the maximum over deleted voters proves the worst-case sensitivity bound.  The average-sensitivity bound is due to $\AS_{m,n}(Q^{(u)})\le\Sens_{m,n}(Q^{(u)})$.

For the final assertion, \Cref{prop:br-uniform-list} supplies absolute constants $K_{\mathsf{BR}}$ and $\varepsilon_{\mathsf{BR}}$ such that every profile has a $K_{\mathsf{BR}}$-list with $\lambda_\sigma(L)\le1-\varepsilon_{\mathsf{BR}}/2$.  Apply the theorem above with $u=1$, $K=K_{\mathsf{BR}}$, and $\gamma=\varepsilon_{\mathsf{BR}}/2$, and set $Q=Q^{(1)}$.  This gives $\Dist(Q)\le 3-7\varepsilon_{\mathsf{BR}}/8$, and the fixed value of $K_{\mathsf{BR}}$ can be absorbed into an absolute constant $C$ in the sensitivity bound.
\end{proof}
We can also prove a worst-case sensitivity bound for deleting multiple voters in the spirit of \citet[Theorem~2.2]{VarmaYoshida}, with additional care needed to control the effect of deletions on the normalization.
\begin{corollary}[$k$-Deletion Worst-Case Sensitivity Bound]\label{cor:k-deletion-worst-case-sensitivity}
Assume the hypothesis of \Cref{thm:gibbs-over-lists}, and let $Q^{(u)}$ be the Gibbs rule defined there.  For $1\le k\le n-1$, define
\[
   \Sens^{(k)}_{m,n}(Q^{(u)})
   \coloneqq
   \sup_{\mathcal E,d}
   \max_{\substack{S\subseteq\Voters\\ |S|=k}}
      \frac{\W_d(Q^{(u)}(\sigma),Q^{(u)}(\sigma^{-S}))}
      {\OPT_{d,\sigma}}
\]
where the supremum ranges over all instances $\mathcal E=(\Voters,\Cands,\sigma)$ with $|\Cands|=m$ and $|\Voters|=n$ and all consistent pseudometrics $d$ with $\OPT_{d,\sigma}>0$, and $\sigma^{-S}$ denotes the profile obtained by deleting all voters in $S$.  Then, we have
\[
   \Sens^{(k)}_{m,n}(Q^{(u)})
   \le
   C_{K,\gamma}(K\log m+u+1)
   \sum_{j=0}^{k-1}\frac{n}{(n-j)^2}.
\]
In particular, if $k\le n/2$, then
\[
   \Sens^{(k)}_{m,n}(Q^{(u)})
   \le
   \frac{4C_{K,\gamma}k(K\log m+u+1)}{n}.
\]
\end{corollary}
\begin{proof}
Fix an instance and a consistent pseudometric $d$ with $\OPT_{d,\sigma}>0$, and fix a $k$-element set $S\subseteq\Voters$.  Choose an arbitrary ordering $S=\{v_1,\ldots,v_k\}$, set $S_j\coloneqq\{v_1,\ldots,v_j\}$ with $S_0=\emptyset$, and write $\sigma_j=\sigma^{-S_j}$ for $j=0,\ldots,k$.  The same pseudometric $d$, restricted to the remaining voters, is consistent with each profile $\sigma_j$.  Let
\[
   D_S
   \coloneqq
   \W_d(Q^{(u)}(\sigma),Q^{(u)}(\sigma^{-S})).
\]
By the triangle inequality for $\W_d$, we have
\[
   D_S
   \le
   \sum_{j=0}^{k-1}\W_d(Q^{(u)}(\sigma_j),Q^{(u)}(\sigma_{j+1})).
\]
For $j<k$, the profile $\sigma_j$ has $n-j\ge2$ voters.  If $\OPT_{d,\sigma_j}=0$, then the zero-optimum convention for finite-distortion rules gives $\W_d(Q^{(u)}(\sigma_j),Q^{(u)}(\sigma_{j+1}))=0$.  Otherwise, by the definition of $\Sens_{m,n-j}$ and the sensitivity bound in \Cref{thm:gibbs-over-lists}, applied to the $(n-j)$-voter profile $\sigma_j$, the inequality in the theorem holds.  In either case, we have
\[
   \W_d(Q^{(u)}(\sigma_j),Q^{(u)}(\sigma_{j+1}))
   \le
   C_{K,\gamma}\frac{K\log m+u+1}{n-j}\OPT_{d,\sigma_j}.
\]
Let $o$ be an optimal candidate for the original profile $\sigma$.  Since $o$ is feasible for each intermediate profile, we have
\[
   \OPT_{d,\sigma_j}
   \le
   \SC_{d,\sigma_j}(o)
   =
   \frac{1}{n-j}\sum_{v\in\Voters\setminus S_j}d(v,o)
   \le
   \frac{1}{n-j}\sum_{v\in\Voters}d(v,o)
   =
   \frac{n}{n-j}\OPT_{d,\sigma}.
\]
Combining the above inequalities and dividing by $\OPT_{d,\sigma}$ yield
\[
   \frac{D_S}{\OPT_{d,\sigma}}
   \le
   C_{K,\gamma}(K\log m+u+1)
   \sum_{j=0}^{k-1}\frac{n}{(n-j)^2}.
\]
Since this inequality holds throughout the range of the supremum and for every $k$-element set $S$, taking the maximum over $S$ and then the supremum proves the first inequality.  If $k\le n/2$, then $n/(n-j)^2\le4/n$ for every $0\le j\le k-1$, which gives the stated simplified bound.
\end{proof}
\subsection{Proofs for Approximate Privacy}\label[appendix]{app:privacy-proofs}
\GibbsPrivacyComparisonLemma*
\begin{proof}
Let $Z$ and $Z'$ be the partition functions for $E$ and $E'$.  Since each energy has a state of value at most $B$, both partition functions are at least $\e^{-\eta B}$.  Hence the probability that $\pi$ assigns to $\Omega_{\le D}^c$, the complement of $\Omega_{\le D}$, is at most $\kappa$, i.e., $\pi(\Omega_{\le D}^c)\le\kappa$, and the same bound holds for $\pi'$.

On $\Omega_{\le D}$, both energies are finite and differ by at most $a\Delta$.  For such $\omega$, the corresponding unnormalized Gibbs weights satisfy
\[
   \frac{\e^{-\eta E'(\omega)}}{\e^{-\eta E(\omega)}}
   =
   \e^{-\eta(E'(\omega)-E(\omega))}
   \in
   [\e^{-\theta},\e^\theta].
\]
Thus the contribution to $Z'$ from $\Omega_{\le D}$ is at most $\e^\theta$ times the contribution to $Z$, while the contribution from $\Omega_{\le D}^c$ is at most $|\Omega|\e^{-\eta D}\le \kappa Z$.  Hence
\[
   Z'
   =
   \sum_{\omega\in\Omega_{\le D}}\e^{-\eta E'(\omega)}
   +
   \sum_{\omega\notin\Omega_{\le D}}\e^{-\eta E'(\omega)}
   \le
   \e^\theta\sum_{\omega\in\Omega_{\le D}}\e^{-\eta E(\omega)}
   +
   |\Omega|\e^{-\eta D}
   \le
   (\e^\theta+\kappa)Z .
\]
For $\omega\in\Omega_{\le D}$, the likelihood ratio obeys
\[
   \frac{\pi(\omega)}{\pi'(\omega)}
   =
   \e^{\eta(E'(\omega)-E(\omega))}\frac{Z'}{Z}
   \le
   \e^{2\theta}+\e^\theta\kappa.
\]
Therefore the low-sublevel part of any event satisfies
\[
   \pi(S\cap\Omega_{\le D})
   \le
   \e^{2\theta}\pi'(S\cap\Omega_{\le D})+\e^\theta\kappa.
\]
Adding the tail bound $\pi(S\cap\Omega_{\le D}^c)\le\pi(\Omega_{\le D}^c)\le\kappa$ gives
\[
   \pi(S)
   \le
   \e^{2\theta}\pi'(S)+(1+\e^\theta)\kappa.
\]
The proof with $E$ and $E'$ interchanged is identical.
\end{proof}
\ApproximatePrivacyUniformListTheorem*
\begin{proof}
If $m=1$, then every list contains only the unique candidate and the released winner is deterministic and independent of the profile, so the distortion and privacy claims are trivial.  Hence assume $m\ge2$.

The distortion argument is the same entropy calculation as in \Cref{thm:gibbs-over-lists}: the inverse temperature $\eta_u$ is at least $(8/\gamma)(K\log m+1)$, so the Gibbs expectation of $\lambda_\sigma(L)$ is at most $1-7\gamma/8$ on every profile.

Fix replacement-neighboring profiles $\sigma$ and $\tau$ on the same candidate set and the same $n$ voters.  Let $E(L)=\lambda_\sigma(L)$ and $E'(L)=\lambda_\tau(L)$; let $\Pi$ and $\Pi'$ be the corresponding Gibbs distributions at inverse temperature $\eta_u$.  Put $B=1-\gamma$, $D=1-\gamma/2$, and $a=A_K(D)$.  Assume $n\ge8a/\gamma$, as required in the theorem statement.
Since the theorem hypothesis gives a $K$-list $L_\sigma$ with $\lambda_\sigma(L_\sigma)\le1-\gamma=B$ and a $K$-list $L_\tau$ with $\lambda_\tau(L_\tau)\le1-\gamma=B$, we have $\min_{L\in\Lists_K}E(L)\le B$ and $\min_{L\in\Lists_K}E'(L)\le B$.  In the notation of \Cref{lem:gibbs-privacy-comparison}, the low-energy set is
\[
   \Omega_{\le D}=\Set*{L\in\Lists_K}{\min(E(L),E'(L))\le D}.
\]
Let $\Lambda(L)\coloneqq\min(E(L),E'(L))$ for $L\in\Omega_{\le D}$.  Then it holds $\Lambda(L)\le D$, hence $A_K(\Lambda(L))\le A_K(D) = a$ and
\[
   \underline R_{K,D}
   =
   \frac{1}{2(KD+1)}
   \le
   \frac{1}{2(K\Lambda(L)+1)}
   =
   \underline R_{K,\Lambda(L)}.
\]
The lower bound on $n$ gives
\[
   \frac1n
   \le
   \frac{\gamma}{8a}
   \le
   \frac{1}{4(KD+1)}
   =
   \frac{\underline R_{K,D}}2
   \le
   \frac{\underline R_{K,\Lambda(L)}}2,
\]
where the middle inequality follows from
$a=4(D+1)(KD+1)$, $D+1\ge1$, and $\gamma<1$.
For each $L\in\Omega_{\le D}$, either $\Lambda(L)=E(L)$ or
$\Lambda(L)=E'(L)$.  The bound above shows that $\Delta=1/n$ satisfies the
smallness condition in \Cref{prop:local-profile-ratio-change} for the
corresponding base energy in either case.  Hence the one-voter replacement
implication applies in the first case with $\sigma$ as the base profile, while
the symmetric implication applies in the second case with $\tau$ as the base
profile.  Since $A_K(\Lambda(L))\le a$, both cases give
\[
   |E(L)-E'(L)|
   \le
   \frac{a}{n}.
\]
Set $T\coloneqq K\log m+u+1$ and $\theta\coloneqq\eta_u a/n$.  Since $\eta_u=(8/\gamma)T$, the lower bound $n\ge8a/\gamma$ gives
\[
   \theta
   =
   \frac{8a}{\gamma}\frac{T}{n}
   \le
   T.
\]
The tail parameter $\kappa$ in \Cref{lem:gibbs-privacy-comparison} satisfies
\[
   \kappa
   =
   |\Lists_K|\e^{-\eta_u(D-B)}
   =
   m^K\e^{-\eta_u\gamma/2}
   =
   m^K\e^{-4T}.
\]
Therefore, we have
\[
   (1+\e^\theta)\kappa
   \le
   (1+\e^T)m^K\e^{-4T}
   =
   (\e^{-4T}+\e^{-3T})\e^{K\log m}
   =
   \e^{-3K\log m-4u-4}
   +
   \e^{-2K\log m-3u-3}
   \le
   \e^{-u},
\]
where the last inequality follows from $K\log m\ge0$ and $\e^{-4u-4}+\e^{-3u-3}\le2\e^{-3u-3}\le\e^{-u}$ for $u\ge1$.  Applying \Cref{lem:gibbs-privacy-comparison} to the list distributions gives, for every event $S\subseteq\Lists_K$, the bound
\[
   \Pi(S)
   \le
   \e^{2\eta_u a/n}\Pi'(S)+\e^{-u}.
\]
The same inequality holds with the two profiles interchanged.  Since $a=A_K(1-\gamma/2)$ and $\eta_u=(8/\gamma)(K\log m+u+1)$, the multiplicative privacy parameter is
\[
   2\theta
    =
    2\frac{\eta_u a}{n}
    =
   \frac{16}{\gamma}A_K\!\left(1-\frac{\gamma}{2}\right)
   \frac{K\log m+u+1}{n}.
\]
Drawing the winner uniformly from the sampled list is the randomized post-processing given by the Markov kernel $L\mapsto p_L$, so the same differential-privacy guarantee holds for the released winner \citep[Proposition~2.1]{DworkRoth2014}.

For the final instantiation, fix $\delta\in(0,1)$ and let $u_\delta\coloneqq\max\{1,\log(1/\delta)\}$.  \Cref{prop:br-uniform-list} supplies the hypothesis of \Cref{thm:gibbs-over-lists} with the absolute constants $K=K_{\mathsf{BR}}$ and $\gamma=\varepsilon_{\mathsf{BR}}/2$.  Define
\[
   C
   \coloneqq
   \frac{32(K_{\mathsf{BR}}+2)}{\varepsilon_{\mathsf{BR}}}
   A_{K_{\mathsf{BR}}}\!\left(1-\frac{\varepsilon_{\mathsf{BR}}}{4}\right).
\]
For $K=K_{\mathsf{BR}}$ and $\gamma=\varepsilon_{\mathsf{BR}}/2$, the lower-bound condition in the privacy statement is $n\ge (16/\varepsilon_{\mathsf{BR}})A_{K_{\mathsf{BR}}}(1-\varepsilon_{\mathsf{BR}}/4)$, which follows from $n\ge C$.  The choice of $u_\delta$ gives $\e^{-u_\delta}\le\delta$: if $\delta\ge\e^{-1}$, then $u_\delta=1$ holds, and otherwise $u_\delta=\log(1/\delta)$ holds.  Moreover, we have
\[
   K_{\mathsf{BR}}\log m+u_\delta+1
   \le
   (K_{\mathsf{BR}}+2)(\log m+\log(1/\delta)+1),
\]
hence
\[
   \varepsilon_{\mathsf{DP}}
   \le
   C\,\frac{\log m+\log(1/\delta)+1}{n}.
\]
The bound $\Dist(Q_\delta)\le 3-\varepsilon$ follows from the distortion bound proved above.
\end{proof}
\section{Finite LP Characterization of the Gibbs Energies}\label[appendix]{sec:computability}
The Gibbs rule over all lists in \cref{sec:gibbs} is not meant to be computationally efficient for large $K$. Still, for fixed~$K$, the energies are finite-state quantities characterized by the following linear programs.  The worst-case distortion value $\rho_\sigma(L)$ of a fixed list can be computed by linear programming, and hence the energy $\lambda_\sigma(L)$ can also be computed through the identity in \Cref{prop:biased-ratio}.  Related LP formulations for randomized social choice distortion have also been used computationally by \citet{FrankLederer2025}.
\begin{proposition}[Finite LP Characterization]\label{prop:lp-computation}
For a fixed profile $\sigma$ and list lottery $p_L$, the value $\rho_\sigma(L)$ can be computed by finitely many linear programs.  In particular, for fixed $K$, all Gibbs energies $\lambda_\sigma(L)$ can be computed by solving $O(m^{K+1})$ polynomial-size LPs or feasibility problems, each with $O((m+n)^2)$ variables and $O((m+n)^3)$ triangle constraints, plus ordinal and normalization constraints.  Once these energies are known, the Gibbs weights can be approximated to arbitrary numerical precision.
\end{proposition}
\begin{proof}
For each candidate $o\in\Cands$, consider variables $d_{xy}$ for $x,y\in\Voters\cup\Cands$ and the linear program $\mathrm{LP}_o$:
\[
\begin{array}{ll}
\text{maximize} & \displaystyle \sum_{c\in\Cands}p_L(c)\frac1n\sum_{v\in\Voters}d_{vc} \\[2mm]
\text{subject to} & d_{xx}=0,\quad d_{xy}=d_{yx}\ge0, \\
& d_{xz}\le d_{xy}+d_{yz}\qquad (x,y,z\in\Voters\cup\Cands),\\
& d_{va}\le d_{vb}\qquad (v\in\Voters,\ a\succ_v b),\\[1mm]
& \displaystyle \frac1n\sum_{v\in\Voters}d_{vo}=1,\\[2mm]
& \displaystyle \frac1n\sum_{v\in\Voters}d_{vc}\ge1\qquad(c\in\Cands).
\end{array}
\]
Every consistent pseudometric $d\consistent\sigma$ with positive optimum can be scaled so that an optimal candidate $o$ has social cost $1$, and then it is feasible for $\mathrm{LP}_o$.  Conversely, every feasible solution is a pseudometric $d\consistent\sigma$ in which $o$ is an optimum of cost $1$, and the objective is exactly the cost ratio of $p_L$.  If $\mathrm{LP}_o$ is infeasible, then no positive-optimum instance with $o$ optimal contributes to the maximum, and this candidate $o$ is ignored.  If any feasible $\mathrm{LP}_o$ is unbounded, then $\rho_\sigma(L)=+\infty$: this is the positive-optimum way in which a list can have infinite distortion, for example when a unanimous profile's non-top candidate can be moved arbitrarily far while the optimum remains normalized to cost $1$.  If all feasible LPs are bounded, their maximum captures all finite positive-optimum ratios.

It remains to detect zero-optimum infinite ratios.  For each $o$, solve the feasibility problem $\mathrm{ZLP}_o$:
\[
\begin{array}{ll}
\text{find} & d \\[1mm]
\text{subject to} & d_{xx}=0,\quad d_{xy}=d_{yx}\ge0, \\
& d_{xz}\le d_{xy}+d_{yz}\qquad (x,y,z\in\Voters\cup\Cands),\\
& d_{va}\le d_{vb}\qquad (v\in\Voters,\ a\succ_v b),\\[1mm]
& \displaystyle \frac1n\sum_{v\in\Voters}d_{vo}=0,\\[2mm]
& \displaystyle \sum_{c\in\Cands}p_L(c)\frac1n\sum_{v\in\Voters}d_{vc}\ge1.
\end{array}
\]
If this feasibility problem is feasible, then an instance with zero optimum and positive lottery cost exists, so $\rho_\sigma(L)=+\infty$ under the zero-denominator convention.  Equivalently, scaling the last inequality shows that the unnormalized objective can be made arbitrarily large while the optimum remains zero.  If $\mathrm{ZLP}_o$ is infeasible, then this candidate $o$ yields no zero-optimum obstruction.  If no $\mathrm{ZLP}_o$ is feasible, zero-optimum instances do not contribute an infinite obstruction, and the maximum of the bounded feasible $\mathrm{LP}_o$ values is $\rho_\sigma(L)$.  For each ordered list there are $m$ programs $\mathrm{LP}_o$ and $m$ programs $\mathrm{ZLP}_o$; since $|\Lists_K|\le m^K$, computing all Gibbs energies requires $O(m^{K+1})$ such polynomial-size LPs or feasibility problems for fixed $K$.  The normalized Gibbs weights are then obtained from these energies by exponentiation and normalization, and can be approximated to any prescribed numerical precision.
\end{proof}
Note that all guarantees in the main theorems concern the exact real-valued Gibbs distribution.  The finite LP characterization shows that the energies and weights can be approximated numerically for fixed $K$, but it does not by itself give a finite-precision sampler with the same distortion or privacy parameters.  Such an implementation would require a separate approximation analysis: the approximation error would have to be charged explicitly to the distortion, sensitivity, and privacy bounds, including control of any probability assigned to lists with infinite energy.  Together with the conservative bound $K_{\mathsf{BR}}\ge10^9K_0$ from \Cref{app:br-consequence}, this makes direct enumeration of $\Lists_{K_{\mathsf{BR}}}$ computationally impractical; the construction is informational rather than algorithmic.
\section{Baseline Results}\label[appendix]{app:baselines}
This appendix proves the auxiliary baseline results used to motivate the main theorem.

For fixed $m$ and $n$, we define the size-fixed distortion:
\[
   \Dist_{m,n}(F)
   =
   \sup_{\substack{\mathcal E=(\Voters,\Cands,\sigma),\ |\Cands|=m,\ |\Voters|=n,\ d\consistent\mathcal E}}
   \frac{\SC_{d,\sigma}(F(\sigma))}{\OPT_{d,\sigma}}.
\]
Here the supremum ranges over all election instances of the given size and all consistent pseudometrics $d$.
Thus statements such as Random Dictatorship having distortion $3-2/n$ are understood as size-fixed bounds, not as global scalar bounds.

For the baseline comparisons and lower bounds, we also use size-fixed average sensitivity~\citep{VarmaYoshida}:
\[
   \AS_{m,n}(F)
   =\!\!\!\!\!\!\!
   \sup_{\substack{\mathcal E=(\Voters,\Cands,\sigma),\ |\Cands|=m,\\  |\Voters|=n,\ d\consistent\mathcal E,\ \OPT_{d,\sigma}>0}}
   \!\!\!
   \frac1n\sum_{v\in\Voters}
   \frac{\W_d(F(\sigma),F(\sigma^{-v}))}{\OPT_{d,\sigma}}.
\]
This is the weaker average-deletion variant of worst-case sensitivity defined in \cref{subsec:sensitivity}.
\subsection{Random Dictatorship}
Let $t_v$ be voter $v$'s top candidate.  Random Dictatorship outputs
\[
   \RD(\sigma)=\frac1n\sum_{v\in\Voters}\delta_{t_v}.
\]
\begin{proposition}[Random Dictatorship]\label{prop:rd}
For every consistent pseudometric instance with $n\ge2$, Random Dictatorship has the distortion bound
\[
   \rho_\sigma(\RD(\sigma))\le3-
   \frac2n.
\]
Moreover, its normalized average deletion distance satisfies
\[
   \frac1n\sum_{i\in\Voters}
   \frac{\W_d(\RD(\sigma),\RD(\sigma^{-i}))}{\OPT_d}
   \le
   \frac4n.
\]
Consequently, it satisfies
\[
   \Dist_{m,n}(\RD)\le3-\frac2n,
   \qquad
   \AS_{m,n}(\RD)\le\frac4n.
\]
However, for every $n\ge2$, its worst-case sensitivity satisfies $\Sens_{2,n}(\RD)\ge1$.
Under replacement adjacency, Random Dictatorship is $(0,1/n)$-differentially private.
\end{proposition}
\begin{proof}
The distortion bound is the Random Dictatorship guarantee of \citet[Theorem~5]{AnshelevichPostl2017}.  In their notation, unrestricted metric preferences correspond to $\alpha=1$, and their bound $2+\alpha-2s/n$, where $s$ is the smallest positive top-choice count, gives $3-2/n$ because $s\ge1$.

For the average-sensitivity bound, let $o$ be an optimal candidate.  After deleting voter $i$, we have
\[
   \RD(\sigma)=\frac1n\delta_{t_i}+\frac{n-1}{n}\RD(\sigma^{-i}).
\]
Couple the common $(n-1)/n$ part identically, and send the remaining $1/n$ mass at $t_i$ to $\RD(\sigma^{-i})$.  This feasible transport plan gives
\[
   \W_d(\RD(\sigma),\RD(\sigma^{-i}))
   \le
   \frac1n\cdot\frac1{n-1}\sum_{j\ne i}d(t_i,t_j).
\]
By the triangle inequality and top-choice property, we have
\[
   d(t_i,t_j)
   \le
   d(t_i,i)+d(i,o)+d(o,j)+d(j,t_j)
   \le
   2d(i,o)+2d(j,o).
\]
Averaging over $i$ gives
\[
   \frac1n\sum_i\W_d(\RD(\sigma),\RD(\sigma^{-i}))
   \le
   \frac{1}{n^2(n-1)}
   \sum_i\sum_{j\ne i}
   \bigl(2d(i,o)+2d(j,o)\bigr)
   =
   \frac4n\OPT_d.
\]
Dividing by $\OPT_d$ proves the average-sensitivity bound.

We next show that Random Dictatorship has $\Omega(1)$ worst-case sensitivity.  Let $a$ and $b$ be two candidates, set $d(a,b)=1$, place $n-1$ voters at $a$ and one voter at $b$, and let each voter rank the co-located candidate first.
Let $v_b$ be the voter located at $b$. Then, we have
\[
   \RD(\sigma)=\frac{n-1}{n}\delta_a+\frac1n\delta_b,
   \qquad
   \RD(\sigma^{-v_b})=\delta_a,
\]
hence $\W_d(\RD(\sigma),\RD(\sigma^{-v_b}))=1/n$ and $\OPT_{d,\sigma}=1/n$.  The normalized Wasserstein distance under deletion of $v_b$ is $1$, so $\Sens_{2,n}(\RD)\ge1$.

It remains to prove the privacy statement.  If $\sigma$ and $\tau$ differ in one voter's ranking, the two Random Dictatorship lotteries differ only in that voter's top choice.  Hence, for every event $S\subseteq\Cands$, the output probabilities satisfy
\[
   \sum_{c\in S}\RD(\sigma)(c)
   \le
   \sum_{c\in S}\RD(\tau)(c)+
   \frac1n.
\]
This is exactly $(0,1/n)$-differential privacy for the released winner.
\end{proof}
\subsection{Kempe's Randomized Rule}
Let $N_c=\abs*{\Set*{v}{t_v=c}}$ and $S=\sum_cN_c^2$.  Define the squared-plurality lottery
\[
   \SQ_c(\sigma)=\frac{N_c^2}{S}.
\]
Kempe's randomized rule \citep{Kempe2020}, which depends only on first-choice counts, is
\[
   F_{\mathsf{Kempe}}(\sigma)=\left(1-
   \frac{1}{m-1}\right)\RD(\sigma)
   +
   \frac{1}{m-1}\SQ(\sigma).
\]
\begin{proposition}[Kempe's Randomized Rule]\label{prop:kempe}
For $m,n\ge2$, Kempe's randomized rule satisfies
\[
   \Dist_{m,n}(F_{\mathsf{Kempe}})
   \le
   3-
   \frac2m,
   \qquad
   \AS_{m,n}(F_{\mathsf{Kempe}})
   \le
   \frac{28}{n}.
\]
However, for every $n\ge2$, its worst-case sensitivity for three-candidate profiles satisfies
\[
   \Sens_{3,n}(F_{\mathsf{Kempe}})
   \ge
   \frac12+\frac{n}{2((n-1)^2+1)}
   =\Omega(1).
\]
Under replacement adjacency, the rule is $(0,\min\{1,9/n\})$-differentially private.
\end{proposition}
\begin{proof}
The distortion bound follows from the guarantee of \citet[Theorem~6.1]{Kempe2020} for this rule; their notation uses $n$ for the number of candidates, which is $m$ here.  It remains to prove the average-sensitivity upper bound, the worst-case sensitivity lower bound, and the differential privacy bound.

Fix a voter $i$ with top candidate $a=t_i$.  Let $N'_c$ and $S'$ denote the counts and squared-count sum after deleting $i$.  Then $N'_a=N_a-1$, all other counts are unchanged, and
\[
   r_a\coloneqq S-S'=2N_a-1.
\]
The squared-plurality lotteries satisfy
\[
   \SQ(\sigma)=\frac{r_a}{S}\delta_a+\frac{S'}{S}\SQ(\sigma^{-i}).
\]
Coupling the common $S'/S$ part identically and sending the remaining $r_a/S$ mass at $a$ to $\SQ(\sigma^{-i})$ gives
\[
   \W_d(\SQ(\sigma),\SQ(\sigma^{-i}))
   \le
   \frac{r_a}{S}\E_{c\sim\SQ(\sigma^{-i})}[d(a,c)].
\]
Let $o$ be an optimal candidate for the original profile.  By the triangle inequality, for each fixed deleted voter $i$ and each candidate $c$, we have
\[
   d(t_i,c)\le d(t_i,o)+d(c,o).
\]
Averaging the resulting fixed-voter bound gives
\[
   \frac1n\sum_i
   \W_d(\SQ(\sigma),\SQ(\sigma^{-i}))
   \le A_1+A_2,
\]
where
\[
   A_1
   \coloneqq
   \frac1n\sum_i\frac{r_{t_i}}{S}d(t_i,o),
   A_2
   \coloneqq
   \frac1n\sum_i\frac{r_{t_i}}{S}
   \E_{c\sim\SQ(\sigma^{-i})}[d(c,o)].
\]
For $A_1$, since $r_{t_i}\le2N_{t_i}$, we have
\[
   A_1
   \le
   \frac{2}{nS}\sum_cN_c^2d(c,o)
   =
   \frac2n\E_{c\sim\SQ(\sigma)}[d(c,o)]
   \le
   \frac{2m}{n}\E_{v\sim\mathrm{Unif}(\Voters)}[d(t_v,o)]
   \le
   \frac{4m}{n}\OPT_d.
\]
Here $S\ge n^2/m$ gives $N_c^2/S\le mN_c/n$, and the last inequality follows from $d(t_v,o)\le d(t_v,v)+d(v,o)\le2d(v,o)$ for every voter $v$, using the triangle inequality and the fact that $t_v$ is voter $v$'s top choice.

For $A_2$, consider each $i$ in the sum separately.  Since $S'\ge (n-1)^2/m$, we have $(N'_c)^2/S'\le mN'_c/(n-1)$ for every candidate $c$, and hence
\[
   \E_{c\sim\SQ(\sigma^{-i})}[d(c,o)]
   \le
   \frac{m}{n-1}\sum_{v\ne i}d(t_v,o)
   \le
   2m\SC_{d,\sigma^{-i}}(o)
   \le
   2m\frac{n}{n-1}\OPT_d,
\]
where the middle inequality uses the same triangle-inequality bound $d(t_v,o)\le2d(v,o)$ for the remaining voters $v\ne i$.
Also, we have
\[
   \frac1n\sum_i\frac{r_{t_i}}{S}
   \le
   \frac{2}{nS}\sum_cN_c^2
   =
   \frac2n.
\]
Hence $A_2\le 4mn/((n-1)n)\OPT_d=4m\OPT_d/(n-1)$.  Therefore the squared-plurality lottery satisfies
\[
   \frac1n\sum_i
   \frac{\W_d(\SQ(\sigma),\SQ(\sigma^{-i}))}{\OPT_d}
   \le
   \frac{4m}{n}+\frac{4m}{n-1}
   \le
   \frac{12m}{n}.
\]
By convexity of Wasserstein distance and \Cref{prop:rd}, we obtain
\[
   \AS_{m,n}(F_{\mathsf{Kempe}})
   \le
   \frac4n+\frac{1}{m-1}\frac{12m}{n}
   \le
   \frac{28}{n},
\]
where we used $m/(m-1)\le2$ for $m\ge2$.

We now give the worst-case sensitivity lower bound.  Let $m=3$ with candidates $a,b,c$, put $c$ at the same location as $a$, and set $d(a,b)=1$.  Let $n-1$ voters have top choice $a$, and let one voter $v_b$ have top choice $b$, with rankings chosen consistently with these locations.  Before deleting $v_b$, the first-choice counts are $N_a=n-1$, $N_b=1$, and $N_c=0$.  Since $m=3$, Kempe's rule gives mass
\[
   \frac{1}{2n}+\frac{1}{2((n-1)^2+1)}
\]
to candidate $b$.  After deleting $v_b$, both the Random Dictatorship and squared-plurality components are supported only at the location of $a$.  The Wasserstein movement is therefore at least the pre-deletion mass on $b$.  Since $\OPT_{d,\sigma}=1/n$, the normalized movement is at least
\[
   \frac12+\frac{n}{2((n-1)^2+1)}.
\]
This proves the $\Omega(1)$ lower bound on worst-case sensitivity.

For differential privacy, let $\sigma$ and $\tau$ differ in one voter's ranking.  The Random Dictatorship component changes by total variation at most $1/n$.  We bound the total variation change of the squared-plurality component.  Let $N_c$ and $N'_c$ be the first-choice counts in the two profiles, and put $u_c=N_c^2$, $u'_c=(N'_c)^2$, $S=\sum_c u_c$, and $S'=\sum_c u'_c$.  The two count vectors differ by moving one unit from some coordinate $a$ to another coordinate $b$, so the affected squared counts satisfy
\[
   N_a^2-(N_a-1)^2\le 2n,
   \qquad
   (N_b+1)^2-N_b^2\le 2n.
\]
All other squared counts are unchanged.  Hence
\[
   \sum_c |u_c-u'_c|\le 4n,
   \qquad
   |S-S'|\le 4n.
\]
Since $S\ge n^2/m$ holds, the $\ell_1$ distance between the squared-plurality lotteries satisfies
\[
   \sum_c\left|\frac{u_c}{S}-\frac{u'_c}{S'}\right|
   \le
   \frac{\sum_c |u_c-u'_c|}{S}
   +
   \left|\frac1S-\frac1{S'}\right|S'
   =
   \frac{\sum_c |u_c-u'_c|+|S-S'|}{S}
   \le
   \frac{8m}{n}.
\]
Thus the squared-plurality component changes by total variation at most $4m/n$.  The mixture coefficient on this component is $1/(m-1)$, and $m\ge2$, so Kempe's rule changes by total variation at most
\[
   \frac1n+\frac{1}{m-1}\frac{4m}{n}
   \le
   \frac9n.
\]
This event-wise additive bound guarantees $(0,\min\{1,9/n\})$-differential privacy for the released winner.
\end{proof}
\section{Lower Bound Proofs}\label[appendix]{app:universal-lower}
We prove the lower-bound rows of \Cref{tab:summary-results}.
For deterministic rules, finite distortion alone implies $\Omega(1)$ worst-case sensitivity, while for even $n$ the two-candidate strict-majority condition gives an $\Omega(1)$ average-sensitivity lower bound.
The universal sensitivity lower bound is proved first for average sensitivity and then transferred to worst-case sensitivity using $\AS\le\Sens$.
We also prove the differential privacy lower bounds for deterministic and randomized rules.

We first state the two-candidate unanimity consequence of finite distortion used throughout the appendix.
\begin{lemma}[Two-Candidate Unanimity from Finite Distortion]\label{lem:finite-distortion-unanimity}
Let $F$ be a possibly randomized ordinal rule with finite metric distortion.  On any two-candidate profile in which every voter ranks $a$ above $b$, the rule $F$ outputs $\delta_a$.  The same statement holds with $a$ and $b$ interchanged.
\end{lemma}
\begin{proof}
Suppose every voter ranks $a$ above $b$.  Place all voters at the location of $a$ and put $b$ at positive distance from $a$.  This metric is consistent with the profile and has optimum cost zero.  If $F$ put positive probability on $b$, then the rule cost would be positive, contradicting finite distortion under the zero-denominator convention in \Cref{sec:prelim}.  Hence $F$ must output $\delta_a$.  The same proof with $a$ and $b$ interchanged gives the symmetric claim.
\end{proof}
\subsection{Deterministic Sensitivity}
We first show the following worst-case sensitivity lower bound for deterministic rules with finite distortion.
\begin{proposition}[Deterministic Worst-Case Sensitivity Lower Bound]\label{prop:det-worst-case-sensitivity-lower}
Fix $n\ge2$.  Let $F$ be a deterministic ordinal rule with finite metric distortion.  Then $\Sens_{2,n}(F)\ge1$ holds.
\end{proposition}
\begin{proof}
Let the two candidates be $a$ and $b$, and fix the voter set $\Voters=\set{1,\ldots,n}$.  For $t=0,1,\ldots,n$, let $\sigma^t$ be the profile in which each voter $i\le t$ ranks $b\succ a$ and each voter $i>t$ ranks $a\succ b$.  By \Cref{lem:finite-distortion-unanimity}, $F(\sigma^0)=\delta_a$ and $F(\sigma^n)=\delta_b$.  Since $F$ is deterministic, there is some $t\in\set{0,\ldots,n-1}$ such that $F(\sigma^t)\ne F(\sigma^{t+1})$.

Let $\widehat\sigma^t$ be the common profile obtained from both $\sigma^t$ and $\sigma^{t+1}$ by deleting voter $t+1$.  Normalize the candidate distance by $d(a,b)=1$.  The triangle inequality for Wasserstein distance gives
\[
   1
   =
   \W_d(F(\sigma^t),F(\sigma^{t+1}))
   \le
   \W_d(F(\sigma^t),F(\widehat\sigma^t))
   +
   \W_d(F(\sigma^{t+1}),F(\widehat\sigma^t)).
\]
Thus one of the two deletion distances is at least $1/2$.  If $t=0$, then $\widehat\sigma^t$ is unanimous for $a$; therefore, \Cref{lem:finite-distortion-unanimity} gives $F(\widehat\sigma^t)=F(\sigma^0)$ and the term involving $\sigma^{t+1}$ is at least $1/2$.  If $t=n-1$, the symmetric argument gives the term involving $\sigma^t$ as the nonzero term.  Hence, in all cases, there is an index $s\in\set{1,\ldots,n-1}$, with $s\in\set{t,t+1}$, such that $(\sigma^s)^{-(t+1)}=\widehat\sigma^t$ and
\[
   \W_d(F(\sigma^s),F((\sigma^s)^{-(t+1)}))\ge\frac12 .
\]
Extend the candidate metric to a consistent metric for $\sigma^s$ by placing the voters who rank $a\succ b$ at the location of $a$ and the voters who rank $b\succ a$ at the location of $b$.  Then, we have
\[
   \OPT_{d,\sigma^s}
   =
   \frac{\min\{s,n-s\}}{n}
   \le
   \frac12,
\]
and this optimum is positive because $1\le s\le n-1$.  Therefore the normalized deletion distance at this instance is at least
\[
   \frac{\W_d(F(\sigma^s),F((\sigma^s)^{-(t+1)}))}{\OPT_{d,\sigma^s}}
   \ge
   1.
\]
Taking the supremum over instances proves the claim.
\end{proof}
We then discuss the average-sensitivity lower bound.
The following proposition shows that for even $n$, every deterministic finite-distortion rule satisfying a natural two-candidate property---selecting the candidate ranked first by a strict majority whenever such a majority exists---has $\Omega(1)$ average sensitivity.
Deterministic rules satisfying the two-candidate property include Copeland~\citep{ABEPS2018}, Plurality Matching~\citep{GHS2020}, and Plurality Veto~\citep{KizilkayaKempe2022}.
\begin{proposition}[Deterministic Strict-Majority Sensitivity Lower Bound]\label{prop:det-majority-sensitivity-lower}
Let $n=2r$ with $r\ge1$.  Let $F$ be a deterministic ordinal rule with finite metric distortion.  Suppose that, on every two-candidate profile in which one candidate is ranked first by a strict majority of voters, the rule selects that candidate.  Then $\AS_{2,n}(F)\ge1$ holds.
\end{proposition}
\begin{proof}
Let the candidates be $a$ and $b$.  Let $\sigma$ be a profile in which a set $S$ of $r$ voters ranks $a\succ b$ and the remaining $r$ voters rank $b\succ a$.  Since $F$ is deterministic, it selects either $a$ or $b$ on this profile.  If $F(\sigma)=\delta_a$, then deleting any voter in $S$ leaves a profile in which $b$ is ranked first by a strict majority, so the output changes to $\delta_b$.  If $F(\sigma)=\delta_b$, the symmetric argument shows that deleting any voter outside $S$ changes the output to $\delta_a$.  Thus at least $r$ voter deletions change the output.

Put the voters in $S$ at the location of $a$ and the remaining voters at the location of $b$, with $d(a,b)=1$.  This metric is consistent with $\sigma$ and satisfies
\[
   \OPT_{d,\sigma}=\frac12.
\]
Each deletion that changes the output contributes Wasserstein distance $1$.  Therefore the normalized average deletion distance at this instance is at least
\[
   \frac1n\cdot
   \frac{r}{\OPT_{d,\sigma}}
   =
   1.
\]
Taking the supremum over instances proves the claim.
\end{proof}
\subsection{Deterministic Differential Privacy}
We give the differential-privacy lower bound for deterministic rules under replacement of one voter's ranking.
\begin{proposition}[Deterministic Privacy Lower Bound]\label{prop:deterministic-privacy-lower}
Fix $n\ge1$.  Let $F$ be a deterministic ordinal rule with finite metric distortion.  If $F$ is $(\varepsilon_{\mathsf{DP}},\delta)$-differentially private under replacement adjacency on all two-candidate $n$-voter profiles, then $\delta\ge1$ holds for any $\varepsilon_{\mathsf{DP}}\ge0$.
\end{proposition}
\begin{proof}
Let the candidates be $a$ and $b$, and fix the voter set $\Voters=\{1,\ldots,n\}$.  For $t=0,1,\ldots,n$, let $\sigma^t$ be the profile in which each voter $i\le t$ ranks $b\succ a$, and each voter $i>t$ ranks $a\succ b$.  For each $t=0,1,\ldots,n-1$, the profiles $\sigma^t$ and $\sigma^{t+1}$ differ only in voter $t+1$, so they are replacement-neighboring.  By \Cref{lem:finite-distortion-unanimity}, $F(\sigma^0)=\delta_a$ and $F(\sigma^n)=\delta_b$.  Since $F$ is deterministic, there is some $t\in\set{0,\ldots,n-1}$ such that $F(\sigma^t)\ne F(\sigma^{t+1})$.

Let $c$ be the candidate selected at $\sigma^t$.  Since $F$ is deterministic and $F(\sigma^t)\ne F(\sigma^{t+1})$, the event $\set{c}$ satisfies
\[
   F(\sigma^t)(c)=1,
   \qquad
   F(\sigma^{t+1})(c)=0.
\]
The differential-privacy inequality for this event and for the neighboring profiles in the order $(\sigma^t,\sigma^{t+1})$ gives
\[
   1
   =
   F(\sigma^t)(c)
   \le
   \e^{\varepsilon_{\mathsf{DP}}}
   F(\sigma^{t+1})(c)
   +
   \delta
   =
   \delta.
\]
This proves the claim.
\end{proof}
\subsection{Universal Sensitivity}
We show the universal sensitivity lower bound.
The proof uses only the unanimity consequence in \Cref{lem:finite-distortion-unanimity}, so it applies to any rule satisfying the stronger distortion requirement in the main theorems.
\begin{proposition}[Universal $1/n$ Lower Bound]\label{prop:universal-one-over-n-lower}
Every randomized ordinal voting rule $F$ with finite metric distortion has average sensitivity at least $1/(2n)$ on two-candidate profiles with $n\ge2$ voters.
Consequently, every such rule also has worst-case sensitivity at least $1/(2n)$ on two-candidate profiles.
\end{proposition}
\begin{proof}
Let $\Voters$ be the voter set with $|\Voters|=n$, and let $\set{a,b}$ be the candidate set.
If $\Voters'\subseteq\Voters$ is nonempty and $T\subseteq\Voters'$, let $\sigma_{\Voters'}^T$ denote the profile over $\Voters'$ in which exactly the voters in $T$ rank $b\succ a$.
Define
\[
   p_{\Voters'}(T)\coloneqq F(\sigma_{\Voters'}^T)(b).
\]
For profiles on the original voter set $\Voters$, write $\sigma^S\coloneqq\sigma_{\Voters}^S$.
\Cref{lem:finite-distortion-unanimity} gives $p_{\Voters'}(\emptyset)=0$ and $p_{\Voters'}(\Voters')=1$ for every nonempty $\Voters'\subseteq\Voters$.

For $i\in \Voters$ and $S\subseteq\Voters$, define the deletion difference
\[
   D_i(S)\coloneqq
   \left|p_{\Voters}(S)-p_{\Voters\setminus\set{i}}(S\setminus\set{i})\right|,
\]
where $S\setminus\set{i}=S$ when $i\notin S$.  This equals the Wasserstein distance between the two output lotteries when the candidate distance is normalized to $d(a,b)=1$.

For a nonunanimous profile $\sigma^S$, put the voters in $S$ at the location of $b$ and the voters in $\Voters\setminus S$ at the location of $a$, with $d(a,b)=1$.  This metric is consistent with the profile and has
\[
   \OPT_{d,\sigma^S}=\frac{\min\{|S|,n-|S|\}}{n}.
\]
Therefore the normalized average deletion distance at this instance is
\begin{equation}\label{eq:lower-instance-sensitivity}
   \frac1n\sum_{i\in\Voters}
   \frac{\W_d(F(\sigma^S),F((\sigma^S)^{-i}))}{\OPT_{d,\sigma^S}}
   =
   \frac{\sum_{i\in\Voters}D_i(S)}{\min\{|S|,n-|S|\}}.
\end{equation}
It remains to show that some nonempty proper subset $S$ makes the right-hand side at least $1/(2n)$.  Suppose, toward a contradiction, that every nonempty proper subset $S\subset \Voters$ satisfies
\begin{equation}\label{eq:lower-contradiction-assumption}
   \sum_{i\in\Voters}D_i(S)<\frac{\min\{|S|,n-|S|\}}{2n}.
\end{equation}
Choose a uniformly random permutation $\pi$ of $\Voters$, and let $S_k=\set*{\pi_1,\ldots,\pi_k}$ for $k=0,\ldots,n$.  Since $p_{\Voters}(S_0)=0$ and $p_{\Voters}(S_n)=1$, the probabilities along this random path satisfy
\[
   1\le\sum_{k=1}^n |p_{\Voters}(S_k)-p_{\Voters}(S_{k-1})|.
\]
For the step from $S_{k-1}$ to $S_k$, the changed voter is $\pi_k$.  Since $S_k\setminus\set{\pi_k}=S_{k-1}$ and $\pi_k\notin S_{k-1}$, deleting $\pi_k$ from either $\sigma^{S_k}$ or $\sigma^{S_{k-1}}$ gives the same profile on $\Voters\setminus\set{\pi_k}$ with $b$-supporter set $S_{k-1}$.  Thus the triangle inequality, with the common term $p_{\Voters\setminus\set{\pi_k}}(S_{k-1})$ inserted between $p_{\Voters}(S_k)$ and $p_{\Voters}(S_{k-1})$, gives
\[
   |p_{\Voters}(S_k)-p_{\Voters}(S_{k-1})|
   \le
   D_{\pi_k}(S_k)+D_{\pi_k}(S_{k-1}).
\]
Taking expectations over the random permutation yields
\[
   1
   \le
   \sum_{k=1}^n
   \E\!\left[D_{\pi_k}(S_k)\right]
   +
   \sum_{k=1}^n
   \E\!\left[D_{\pi_k}(S_{k-1})\right].
\]
For the first sum, condition on the realized set $S_k=S$.  Given this event, the first $k$ positions form a uniformly random ordering of $S$, so $\pi_k$ is uniform on $S$.  Hence, when $1\le k\le n-1$, the set $S_k$ is a nonempty proper subset of $\Voters$, and \eqref{eq:lower-contradiction-assumption} gives
\[
   \E\!\left[D_{\pi_k}(S_k) \mid S_k\right]
   =
   \frac1k\sum_{i\in S_k}D_i(S_k)
   \le
   \frac1k\sum_{i\in\Voters}D_i(S_k)
   <
   \frac{\min\{k,n-k\}}{2nk}
   \le
   \frac1{2n}.
\]
The term $k=n$ is zero because both the all-$b$ profile and every all-$b$ deletion profile output $b$ with probability one.  Hence the first expected sum is less than $(n-1)/(2n)$.

For the second sum, condition on $S_{k-1}$.  When $2\le k\le n$, the voter $\pi_k$ is uniform in $\Voters\setminus S_{k-1}$, and \eqref{eq:lower-contradiction-assumption} gives
\[
   \E\!\left[D_{\pi_k}(S_{k-1})\mid S_{k-1}\right]
   =
   \frac1{n-k+1}
   \sum_{i\in\Voters\setminus S_{k-1}}D_i(S_{k-1})
   \le
   \frac1{n-k+1}
   \sum_{i\in\Voters}D_i(S_{k-1})
   <
   \frac1{2n}.
\]
The term $k=1$ is zero because both the all-$a$ profile and every all-$a$ deletion profile output $a$ with probability one.  Hence the second expected sum is also less than $(n-1)/(2n)$.  Combining the two estimates gives
\[
   1<\frac{n-1}{n},
\]
a contradiction.  Therefore some nonempty proper subset $S$ satisfies
\[
   \sum_{i\in\Voters}D_i(S)
   \ge
   \frac{\min\{|S|,n-|S|\}}{2n}.
\]
For this two-candidate instance, \eqref{eq:lower-instance-sensitivity} gives normalized average deletion distance at least $1/(2n)$.  Taking the supremum over instances proves the average-sensitivity lower bound.  Since $\AS_{2,n}(F)\le\Sens_{2,n}(F)$, the same lower bound applies to worst-case sensitivity.
\end{proof}
\subsection{Universal Approximate Differential Privacy}
We give the universal lower bound for approximate differential privacy under replacement of one voter's ranking.
\begin{proposition}[Universal Approximate Differential Privacy Lower Bound]\label{prop:privacy-lower}
Fix $n\ge1$.  Let $F$ be a randomized ordinal rule with finite metric distortion.  If $F$ is $(\varepsilon_{\mathsf{DP}},\delta)$-differentially private under replacement adjacency on all two-candidate $n$-voter profiles, then it holds that
\[
   \delta
   \ge
   \left(\sum_{j=0}^{n-1} \e^{j\varepsilon_{\mathsf{DP}}}\right)^{-1}.
\]
In particular, if $\varepsilon_{\mathsf{DP}}\le c/n$ for a constant $c\ge0$, then we have $\delta\ge (\e^{-c})/n$; hence $\varepsilon_{\mathsf{DP}}=O(1/n)$ forces $\delta=\Omega(1/n)$.
\end{proposition}
\begin{proof}
Let the candidates be $a$ and $b$, and fix the voter set $\Voters=\{1,\ldots,n\}$.  For $t=0,1,\ldots,n$, let $\sigma^t$ be the profile in which each voter $i\le t$ ranks $b\succ a$, and each voter $i>t$ ranks $a\succ b$.  For each $t=0,1,\ldots,n-1$, the profiles $\sigma^t$ and $\sigma^{t+1}$ differ only in voter $t+1$, so they are replacement-neighboring.  By \Cref{lem:finite-distortion-unanimity}, $F(\sigma^0)=\delta_a$ and $F(\sigma^n)=\delta_b$.

Put $p_t\coloneqq F(\sigma^t)(b)$.  Then $p_0=0$ and $p_n=1$.  Applying approximate differential privacy to the released-winner event $\set{b}$ for the neighboring pair $(\sigma^{t+1},\sigma^t)$ gives
\[
   p_{t+1}\le \e^{\varepsilon_{\mathsf{DP}}}p_t+\delta.
\]
Iterating this recursion from $p_0=0$ gives
\[
   p_n\le \delta\sum_{j=0}^{n-1}\e^{j\varepsilon_{\mathsf{DP}}}.
\]
Since $p_n=1$, the claimed lower bound on $\delta$ follows.
If $\varepsilon_{\mathsf{DP}}\le c/n$, then $\sum_{j=0}^{n-1}\e^{j\varepsilon_{\mathsf{DP}}}\le n\e^c$, which gives the stated $(\e^{-c})/n$ lower bound.
\end{proof}
\section{Additional Related Work}\label[appendix]{app:related-work}
\paragraph{Sensitivity for randomized algorithms and metric voting.}
\citet{VarmaYoshida} initiated a systematic study of average sensitivity for graph algorithms, measuring the average movement of an output distribution under uniformly random deletion of one input element.  Subsequent work has studied average sensitivity in spectral clustering~\citep{PengYoshida2020}, Euclidean $k$-clustering~\citep{YoshidaIto2022}, dynamic programming~\citep{KumabeYoshidaDP}, knapsack problems~\citep{KumabeYoshidaKnapsack}, decision-tree learning~\citep{HaraYoshida2023}, hierarchical $k$-median clustering~\citep{LiHeBaiPeng2025}, and geometric algorithms~\citep{EbbensYoshida2026}.  The main stability notion in this paper follows the stronger worst-case deletion-sensitivity viewpoint of \citet{YoshidaZhou2021}, who study edge and vertex deletion for maximum matching, with deterministic changes measured by symmetric difference and randomized changes measured by Wasserstein distance over matchings.  The closest Gibbs-distribution analogue is \citet{YoshidaZhang2026}, who obtain low-sensitivity matching algorithms by sampling from Gibbs distributions over matchings.
We adapt this worst-case viewpoint to metric voting: the deleted object is a voter, the Wasserstein ground distance is given by the candidate metric, and the normalization is by the original metric-voting optimum.  Interestingly, for fixed $m$, our results exhibit a stronger deterministic--randomized contrast than the one usually visible from existing sensitivity results.  Specifically, the Gibbs rule has $O(1/n)$ worst-case sensitivity, whereas \cref{prop:det-worst-case-sensitivity-lower} gives $\Omega(1)$ worst-case sensitivity for finite-distortion deterministic rules.  In contrast, lower bounds in the sensitivity literature often apply to randomized algorithms as well, while the worst-case-sensitivity separation of \citet{YoshidaZhou2021} for maximum matching under edge deletion is between $O_\varepsilon(1)$ for randomized algorithms and $\Omega(\log^* n)$ for deterministic algorithms, where $\varepsilon$ is the approximation parameter and $n$ is the number of vertices.  In light of this, metric voting exhibits a stark deterministic--randomized contrast already for two candidates: randomization achieves the one-voter deletion scale, whereas finite-distortion deterministic rules remain at a constant worst-case scale.  Within metric voting, \citet{BergerFeldmanGkatzelisTan2024} study a different robustness axis in the learning-augmented framework, where the performance guarantee depends on the error of a predicted optimal candidate rather than on deletion of one voter.
\paragraph{Differential privacy in mechanism design and social choice.}
Differential privacy was introduced by \citet{DworkMcSherryNissimSmith2006} as an event-wise bound on how much one individual's data can affect output probabilities, with the approximate version treated in standard references such as \citet{DworkRoth2014}.  It has also been used as a design principle in mechanism design, beginning with \citet{McSherryTalwar2007} and including approximately optimal mechanism design via differential privacy~\citep{NissimSmorodinskyTennenholtz2012}.  In voting, \citet{LiuLuXiaZikas2020} analyze deterministic voting rules under distributional differential privacy, where privacy is defined relative to uncertainty in an adversarial observer's prior over individual votes; \citet{LiLiuXiaCaoWang2023} design randomized Condorcet-style rules satisfying standard differential privacy together with voting axioms.  Other related work has studied privacy-preserving randomized social choice, local differential privacy as a defense against electoral control by voter deletion, and private rank aggregation~\citep{Torra2020,TaoChenXuShi2022,AlabiGhaziKumarManurangsi2022}. These papers address privacy and preference aggregation from different objectives and models; our focus is the compatibility of approximate differential privacy for a sampled winner under replacement adjacency with a constant improvement over the metric-distortion barrier of $3$.
\paragraph{Metric distortion.}
The distortion framework of \citet{ProcacciaRosenschein2006} measures the loss caused by making social-choice decisions from ordinal rankings rather than the underlying utilities or costs.  In the metric version developed by \citet{ABEPS2018}, voters and candidates lie in an unknown metric space, the social cost of a candidate is its total distance to the voters, and the voting rule sees only the induced rankings.  The early metric-distortion results already identified the central deterministic benchmark: \citet{ABEPS2018} proved a lower bound of $3$ and gave constant-distortion guarantees for familiar deterministic rules, including a distortion-$5$ guarantee for Copeland.  Subsequent work studied deterministic rules from several angles.
\citet{GoelKrishnaswamyMunagala2017} studied lower bounds and fairness properties for deterministic rules.
\citet{MunagalaWang2019} improved the best deterministic upper bound to $2+\sqrt 5$, and \citet{GHS2020} resolved the deterministic case by giving a rule with distortion $3$, matching the lower bound.
Later work gave simpler optimal deterministic rules and further variants~\citep{KizilkayaKempe2022}.
See the survey of \citet{AFSV2021} for an overview of metric distortion and related variants in social choice.

Randomization changes the question from the optimal deterministic constant to whether the barrier $3$ can be beaten in the worst case.  Random Dictatorship is the canonical randomized baseline; for $n$ voters, it has distortion at most $3-2/n$ in the metric setting~\citep{AnshelevichPostl2017}.  Other randomized rules, including the rule of \citet{Kempe2020}, give bounds such as $3-2/m$ for $m$ candidates, but these guarantees still approach $3$ when the number of voters or candidates is allowed to grow.  These results left open whether distortion $2$, the general lower bound for randomized rules in metric spaces~\citep{AnshelevichPostl2017}, was achievable.  This possibility was ruled out by lower bounds above $2$ due to \citet{CharikarRamakrishnan2022} and \citet{PulyassarySwamy2021}.  On the upper-bound side, \citet{CRWW2024} broke the randomized $3$ barrier by combining Maximal Lotteries with complementary rules and proving distortion at most $2.753$.  \citet{Cai2026} then showed that bounded randomness already suffices to beat $3$: at each profile, one can choose a constant-size list of candidates so that the uniform lottery on that list has distortion below $3$ by an absolute margin.  Our construction takes this bounded-randomness existence theorem as the starting point and asks whether a constant improvement over $3$ can be made compatible with worst-case sensitivity and approximate differential privacy.
\section{Limitations and Further Directions}\label[appendix]{app:limitations}
This appendix collects limitations and further directions that are useful for interpreting the main guarantees.
\paragraph{The Logarithmic Dependence on the Number of Candidates.}
The factor $K\log m$ comes from the entropy cost of randomizing over all $K$-lists.  The bounded-randomness theorem guarantees the existence of one $K$-list with $\lambda_\sigma(L)\le B$ among at most $m^K$ list states, but it does not give a rule with controlled sensitivity that identifies that list or a lower bound on the number of lists satisfying the same inequality.  The following two-level example illustrates why the entropy term appears in this analysis.  Suppose that the feasible lists form a two-level energy landscape: one list has energy $B$, and the remaining $m^K-1$ lists have energy $B+\xi$, where $\xi>0$ is fixed.  Under Gibbs weights with inverse temperature $\eta$, the excess expected energy satisfies
\[
\E_\eta[E]-B
=
\xi\frac{(m^K-1)\e^{-\eta \xi}}
{1+(m^K-1)\e^{-\eta \xi}}.
\]
Thus, making the expected energy at most $B+\xi/2$ requires
\[
\eta\ge \frac{1}{\xi}\log(m^K-1)
=
\Omega(K\log m).
\]
This example captures a limitation of an analysis that only controls the number of feasible lists and their energy gap from the minimum.  Still, the example leaves room for sharper bounds that use additional structure of feasible lists or of the distortion objective, such as a selection procedure with controlled sensitivity for a list with $\lambda_\sigma(L)\le B$, a proof that the sublevel set $\Set*{L}{\lambda_\sigma(L)\le B+O(1)}$ has large measure, or a different distribution over lists that avoids paying entropy over all $m^K$ states.
\paragraph{Sensitivity and Differential Privacy.}
Worst-case sensitivity and differential privacy control different effects of changing one voter.  Sensitivity bounds normalized Wasserstein movement of the candidate lottery under deletion, while differential privacy bounds event probabilities for the released winner under replacement.  The former is a metric transport bound, whereas the latter is an event-wise probability bound.  The two definitions alone do not imply either direction; an implication would need additional assumptions connecting candidate distances and output events.  We thus prove the two guarantees separately, although both proofs rely on the same low-ratio stability estimates for fixed lists.  Clarifying the exact relationship between these guarantees, or developing a unified stability notion that captures both, remains open.  \Cref{prop:privacy-lower} shows that, when $\varepsilon_{\mathsf{DP}}=O(1/n)$, an additive privacy parameter of order $1/n$ is unavoidable in general.  If $\delta=\Theta(1/n)$, our upper bound gives $\varepsilon_{\mathsf{DP}}=O((\log m+\log n)/n)$, while the lower bound applies in the stricter regime $\varepsilon_{\mathsf{DP}}=O(1/n)$; closing this logarithmic gap also remains open.
\paragraph{Constants and Computation.}
The constant inherited from bounded randomness is extremely small, and rational rounding can make the list size enormous.  The rule is therefore best viewed as an existential construction.  For fixed $K$, the Gibbs energies can be computed by finitely many linear programs, and the resulting Gibbs weights can be approximated to arbitrary numerical precision; \Cref{sec:computability} gives the formulation.  Improving the bounded-randomness constants, reducing the list size, or finding a more compact representation of the biased-metric ratio value are natural open problems.
\end{document}